\documentclass[a4paper,UKenglish]{lipics-v2019}

\usepackage{stmaryrd}%
\usepackage{amssymb}%

\usepackage{thmtools,thm-restate}%
\theoremstyle{plain}

\usepackage{tabu} 

\newcommand{\lang}[1]{{\mathcal{L}(#1)}}%
\newcommand{\preceqq}{\preccurlyeq}

\newcommand{\wt}{\widetilde}

\newcommand{\cC}{\mathcal{C}}

\newcommand{\cF}[1]{\mathsf{Can}^{#1}}
\newcommand{\cG}[1]{\mathsf{Res}^{#1}}
\newcommand{\cL}{\mathsf{R}}
\newcommand{\cH}{\mathsf{H}}

\newcommand{\cM}{\mathcal{M}}
\newcommand{\cN}{\mathcal{N}}

\newcommand{\cP}{\mathcal{P}}

\newcommand{\cR}{\mathcal{R}}
\newcommand{\cS}{\mathcal{S}}
\newcommand{\cT}{\mathcal{T}}

\newcommand{\rr}{\sim^{r}}
\newcommand{\qo}{\preceqq}
\newcommand{\qon}{\prec}
\newcommand{\ql}{\preceqq^{\ell}}
\newcommand{\qr}{\preceqq^{r}}
\newcommand{\qln}{\prec^{\ell}}
\newcommand{\qrn}{\prec^{r}}

\newcommand{\rrN}{\rr_{\cN}}

\newcommand{\rrL}{\rr_{L}}
\newcommand{\qlN}{\ql_{\cN}}
\newcommand{\qrN}{\qr_{\cN}}

\newcommand{\qlL}{\ql_{L}}
\newcommand{\qrL}{\qr_{L}}

\newcommand{\ud}{\stackrel{\rm\scriptscriptstyle def}{=}}
\newcommand{\udiff}{\stackrel{\rm\scriptscriptstyle def}{\iff}}

\newcommand{\Lra}{\Leftrightarrow}
\newcommand{\Ra}{\Rightarrow}
\newcommand{\La}{\Leftarrow}
\newcommand{\ra}{\rightarrow}

\DeclareMathOperator{\cl}{{cl}}

\DeclareMathOperator{\pre}{pre}
\DeclareMathOperator{\post}{post}
\DeclareMathOperator{\Suf}{\cS}
\DeclareMathOperator{\Pref}{\cP}

\DeclareMathOperator{\row}{r}
\DeclareMathOperator{\Rows}{Rows}

\newcommand{\eox}{\hfill{\ensuremath{\Diamond}}}

\newcommand{\qA}{\qr_{L_{\Suf}}}
\newcommand{\qAn}{\qrn_{L_{\Suf}}}

\usepackage{tikz}
\usetikzlibrary{arrows}
\usetikzlibrary{positioning}
\usetikzlibrary{automata}
\usetikzlibrary{shapes,shapes.geometric}
\usepackage{tikz-cd}
\usetikzlibrary{decorations.pathmorphing}

\usepackage[utf8]{inputenc}
\usepackage{newunicodechar}
\newunicodechar{ε}{\varepsilon}
\newunicodechar{⟨}{\langle}
\newunicodechar{⟩}{\rangle}
\newunicodechar{⟦}{\llbracket}
\newunicodechar{⟧}{\rrbracket}
\newunicodechar{⌈}{\lceil}
\newunicodechar{⌉}{\rceil}
\newunicodechar{⇃}{\mathclose{\downharpoonleft}}
\newunicodechar{⋯}{\cdots}
\newunicodechar{∈}{\in}
\newunicodechar{∉}{\notin}
\newunicodechar{⊕}{\oplus}
\newunicodechar{⊗}{\otimes}
\newunicodechar{∘}{\circ}
\newunicodechar{α}{\alpha}
\newunicodechar{δ}{\delta}
\newunicodechar{φ}{\varphi}
\newunicodechar{ψ}{\psi}
\newunicodechar{μ}{\mu}
\newunicodechar{π}{\pi}
\newunicodechar{ρ}{\rho}
\newunicodechar{τ}{\tau}
\newunicodechar{Σ}{\Sigma}
\newunicodechar{≤}{\leq}
\newunicodechar{≥}{\geq}
\newunicodechar{≝}{\stackrel{\rm\scriptscriptstyle def}{=}}
\newunicodechar{𝒪}{\mathcal{O}}
\newunicodechar{𝒯}{\mathcal{T}}
\newunicodechar{𝒴}{\mathcal{Y}}
\newunicodechar{Δ}{\Delta}

\newcommand{\len}[1]{{\vert{#1}\vert}}

\usepackage{amsmath}
\usepackage{amssymb}
\usepackage{stmaryrd}
\usepackage{mathtools}

\allowdisplaybreaks %

\usepackage[ruled,noend,noline,linesnumbered]{algorithm2e}
\makeatletter
\newcommand{\RemoveAlgoNumber}{\renewcommand{\fnum@algocf}{\AlCapSty{\AlCapFnt\algorithmcfname}}}
\makeatother

\nolinenumbers

\title{A Quasiorder-based Perspective on Residual Automata}
\titlerunning{A Quasiorder-based Perspective on Residual Automata}%

\author{Pierre Ganty}{IMDEA Software Institute, Madrid, Spain}{pierre.ganty@imdea.org}{0000-0002-3625-6003}{Partially supported by the Madrid regional project S2018/TCS-4339 BLOQUES and the Ramón y Cajal fellowship RYC-2016-20281.}

\author{Elena Gutiérrez}{IMDEA Software Institute, Madrid, Spain\\ Universidad Politécnica de Madrid, Spain}{elena.gutierrez@imdea.org}{0000-0001-5999-7608}{
Partially supported by the BES-2016-077136 grant from the Spanish Ministry of Economy, Industry and Competitiveness.}

\author{Pedro Valero}{IMDEA Software Institute, Madrid, Spain\\ Universidad Politécnica de Madrid, Spain}{pedro.valero.mejia@gmail.com}{0000-0001-7531-6374}{}

\authorrunning{P. Ganty and E. Gutiérrez and P. Valero}%

\Copyright{Pierre Ganty and Elena Gutiérrez and Pedro Valero}%



\category{}%

\relatedversion{}%
\supplement{}

\funding{All authors were partially supported by the Spanish project PGC2018-102210-B-I00 BOSCO}%
\hideLIPIcs

\begin{document}

\maketitle

\begin{abstract}
In this work, we define a framework of automata constructions based on quasiorders over words to provide new insights on the class of residual automata.
We present a new residualization operation and a generalized double-reversal method for building the canonical residual automaton for a given language.
Finally, we use our framework to offer a quasiorder-based perspective on NL\(^*\), an online learning algorithm for residual automata.
We conclude that quasiorders are fundamental to residual automata as congruences are to deterministic automata.
\end{abstract}

\section{Introduction}

Residual automata (RFAs for short) are finite-state automata for which each state defines a \emph{residual} of its language, where the residual of a language \(L\) by a word \(u\) is defined as the set of words \(w\) such that \(uw \in L\).
The class of RFAs lies between deterministic (DFAs) and nondeterministic automata (NFAs).
They share with DFAs a significant property: the existence of a canonical minimal form for any regular language.
On the other hand, they share with NFAs the existence of automata that are exponentially smaller (in the number of states) than the corresponding minimal DFA for the language.
These properties make RFAs specially appealing in certain areas of computer science such as Grammatical Inference~\cite{denis2004learning,Kasprzik2011Inference}.

RFAs were first introduced by Denis et al.~\cite{denis2000residual,denis2002residual}.
They defined an algorithm for \emph{residualizing} an automaton, which is a variation of the well-known subset construction used for determinization, and showed that there exists a \emph{unique} \emph{canonical} RFA, which is minimal in the number of states, for every regular language.
Moreover, they showed that the residual-equivalent of the double-reversal method~\cite{brzozowski1962canonical} holds, i.e.\ residualizing an automaton \(\cN\) whose reverse is residual yields the canonical RFA for the language accepted by \(\cN\).

Later, Tamm~\cite{tamm2015generalization} generalized the double-reversal method for RFAs by giving a sufficient and necessary condition that guarantees that the residualization operation defined by Denis et al.~\cite{denis2002residual} yields the canonical RFA.
In fact, this generalization comes in the same lines as that of Brzozowski and Tamm~\cite{brzozowski2014theory} for the double-reversal method for building the minimal DFA.

These results evidence the existence of a relationship between RFAs and DFAs.
In fact, a connection between these two classes of automata was already established by Myers et al.~\cite{AdamekMUM14,MyersAMU15} from a category-theoretical point of view.
Concretely, they~\cite{AdamekMUM14} use this perspective to address the residual-equivalent of the double-reversal method proposed by Denis et al.~\cite{denis2002residual} to obtain the canonical RFA.

In this work we evidence this connection between RFAs and DFAs from the point of view of quasiorders over words.
Specifically, we show that \emph{quasiorders} are fundamental to RFAs as \emph{congruences} are for DFAs.

Previously, we studied the problem of building DFAs using congruences, i.e., equivalence relations over words with good properties w.r.t. concatenation~\cite{ganty2019congruence}.
This way, we derived several well-known results about minimization of DFAs, including the double-reversal method and its generalization by Brzozowski and Tamm~\cite{brzozowski2014theory}.
While the use of congruences over words suited for the construction of a subclass of residual automata, namely, \emph{deterministic} automata, these are no longer useful to describe the more general class of \emph{nondeterministic} residual automata.
By moving from \emph{congruences} over words to \emph{quasiorders}, we are able to introduce nondeterminism in our automata constructions.

We consider quasiorders with good properties w.r.t. \emph{right} and \emph{left} concatenation.
In particular, we define the so-called right \emph{language-based} quasiorder, whose definition relies on a given regular language; and the right \emph{automata-based} quasiorder, whose definition relies on a finite representation of the language, i.e., an automaton.
We also give counterpart definitions for quasiorders that behave well with respect to \emph{left} concatenation.
Relying on quasiorders that \emph{preserve} a given regular language, i.e., the closure of the language w.r.t. the quasiorder coincides with the language, we will provide a framework of finite-state automata constructions for the language.

When instantiating our automata constructions using the right language-based quasiorder, we obtain the canonical RFA for the given language; while using the right automata-based quasiorder yields an RFA for the language accepted by the automaton that has, at most, as many states as the RFA obtained by the residualization operation defined by Denis et al.~\cite{denis2002residual}.
Similarly, \emph{left} automata-based and language-based quasiorders yield co-residual automata, i.e., automata whose reverse is residual.

Our quasiorder-based framework allows us to give a simple correctness proof of the double-reversal method for building the canonical RFA.
Moreover, it allows us to generalize this method in the same fashion as Brzozowski and Tamm~\cite{brzozowski2014theory} generalized the double-reversal method for building the minimal DFA.
Specifically, we give a characterization of the class of automata for which our automata-based quasiorder construction yields the canonical RFA.

We compare our characterization with the class of automata, defined by Tamm~\cite{tamm2015generalization}, for which the residualization operation of Denis et al.~\cite{denis2002residual} yields the canonical RFA and show that her class of automata is strictly contained in the class we define.
Furthermore, we highlight the connection between the generalization of Brzozowski and Tamm~\cite{brzozowski2014theory} and the one of Tamm~\cite{tamm2015generalization} for the double-reversal methods for DFAs and RFAs, respectively.

Finally, we revisit the problem of learning residual automata from a quasiorder-based perspective.
Specifically, we observe that the NL\(^*\) algorithm defined by Bollig et al.~\cite{bollig2009angluin}, inspired by the popular Angluin's L\(^*\) algorithm for learning DFAs~\cite{angluin1987learning}, can be seen as an algorithm that starts from a quasiorder and refines it at each iteration.
At the end of each iteration, the automaton built by NL\(^*\) coincides with our quasiorder-based automata construction applied to the refined quasiorder.

\medskip\noindent\textbf{Structure of the paper.}%
\label{subp:structure}
After preliminaries in Section~\ref{sec:preliminaries}, we introduce in Section~\ref{sec:automataConstructions} automata constructions based on quasiorders and establish the duality between these constructions when using right and left quasiorders.
We instantiate these constructions in Section~\ref{sec:Instantiation} with the language-based and automata-based quasiorders and study the relations between the resulting automata.
As a consequence, we derive in Section~\ref{sec:Novel} a generalization of the double-reversal method for building the canonical RFA for a language.
In addition, we show a novel quasiorder-based perspective on the NL\(^*\) algorithm for learning residual automata in Section~\ref{sec:LearningNL:qo}.
Finally, Appendix~\ref{sec:LearningNL} includes a formal description of the NL\(^*\) algorithm, Appendix~\ref{sec:supp-results} is dedicated to supplementary results, including the pseudocode of our quasiorder-based version of NL\(^*\), and Appendix~\ref{sec:proofs} contains all the deferred proofs.

\section{Preliminaries}%
\label{sec:preliminaries}

\noindent\textbf{Languages.}
Let \(\Sigma\) be a finite nonempty \emph{alphabet} of symbols.
Given a word \(w \in \Sigma^*\), we will use \(|w|\) to denote the \emph{length} of \(w\).
We denote \(w^R\) the \emph{reverse} of \(w\).
Given a language \(L \subseteq \Sigma^*\), \(L^R ≝ \{w^R \mid w \in L\}\) denotes the \emph{reverse language} of \(L\) and \(L^c\), its \emph{complement} language.

We denote the left (resp. right) quotient of \(L\) by a word \(u\), also known as \emph{residual}, as \(u^{-1}L \ud \{w \in Σ^* \mid uw \in L\}\) (resp. \(Lu^{-1} \ud \{w \in Σ^* \mid wu \in L\}\)).
Denis et al.~\cite{denis2002residual} defined the notion of \emph{composite} and \emph{prime} residuals that we extend to right quotients as follows.
A left (resp.\ right) quotient \(u^{-1}L\) (resp. \(Lu^{-1}\)) is \emph{composite} if{}f it is the union of all the left (resp.\ right) quotients that it strictly contains, i.e.\ \(u^{-1}L = \bigcup_{x \in Σ^*, \; x^{-1}L \subsetneq u^{-1}L} x^{-1}L\) (resp.\ \(Lu^{-1} = \bigcup_{x \in Σ^*, \; Lx^{-1} \subsetneq Lu^{-1}} Lx^{-1}\)).
Otherwise, we say the quotient is \emph{prime}.

\medskip\noindent\textbf{Automata.}
A  \emph{(nondeterministic) finite-state automaton} (NFA for short), or simply \emph{automaton}, is a 5-tuple \(\cN = (Q, \Sigma, \delta, I, F)\), where \(Q\) is a finite set of \emph{states}, \(\Sigma\) is an alphabet, \({I\subseteq Q}\) are the \emph{initial} states, \(F \subseteq Q\) are the \emph{final} states, and \(\delta: Q \times \Sigma \ra \wp(Q)\) is the \emph{transition} function, where \(\wp(Q)\) denotes the powerset w.r.t. \(Q\).
We denote the \emph{extended transition function} from \(\Sigma\) to \(\Sigma^*\) by \(\hat{\delta}\), defined in the usual way, and, given \(w \in Σ^*\) and \(S \in \wp(Q)\), we define \(\post_w^{\cN}(S) \ud \{q \in Q \mid \exists q' \in S, \; q \in \hat{δ}(q',w)\}\) and \(\pre_w^{\cN}(S) \ud \{q \in Q \mid \exists q' \in S, \; q' \in \hat{δ}(q,w)\}\).

Given \(S,T \subseteq Q\), \(W^{\cN}_{S,T} \ud \{w \in \Sigma^* \mid \exists q \in S, q' \in T,\; q' \in \hat{\delta}(q,w)\}\).
In particular, when \(S = \{q\}\) and \(T = F\), we say that \(W^{\cN}_{q,F}\) is \emph{the right language} of state \(q\).
Likewise, when \(S = I\) and \(T = \{q\}\), we say that \(W^{\cN}_{I,q}\) is the \emph{left language} of state \(q\).
In general, we omit the automaton \(\cN\) from the superscript when it is clear from the context.
We say that a state \(q\) is \emph{unreachable} if{}f \(W^{\cN}_{I,q} = \varnothing\) and we say that \(q\) is \emph{empty} if{}f \(W^{\cN}_{q,F} = \varnothing\).
Finally, the language accepted by an automaton \(\cN\) is \(\lang{\cN} = \bigcup_{q \in I} W_{q,F}^{\cN} = \bigcup_{q \in F} W_{I,q}^{\cN} = W_{I,F}^{\cN}\).

The NFA \(\cN' = (Q', Σ, δ', I', F')\) is a \emph{sub-automaton} of \(\cN\) if{}f \(Q' \subseteq Q\), \(I' \subseteq I\), \(F' \subseteq F\) and \(q' \in δ'(q,a) \Ra q' \in δ(q,a)\) with \(q,q' \in Q\) and \(a \in \Sigma\).
The \emph{reverse} of \(\cN\), denoted by \(\cN^R\), is defined as \(\cN^R = (Q, \Sigma, \delta_r, F, I)\) where \(q \in \delta_r (q',a)\) if{}f \(q' \in \delta(q,a)\).
Clearly, \(\lang{\cN}^R = \lang{\cN^R}\).

\medskip\noindent\textbf{Residual Automata.} A \emph{residual finite-state automaton} (RFA for short) is an NFA such that the right language of each state is a left quotient of the accepted language.
We write RFA instead of RFSA~\cite{denis2002residual} to be consistent with the abbreviations NFA and DFA.
Formally, an RFA is an automaton \(\cN = (Q, Σ, δ, I, F)\) such that \(\forall q \in Q, \exists u \in Σ^*, \; W^{\cN}_{q,F} = u^{-1}\lang{\cN}\).

We say an automaton is \emph{co-residual} (co-RFA for short) if its reverse is an RFA, i.e., \(\forall q \in Q, \exists u \in Σ^*, \; W^{\cN}_{I,q} = \lang{\cN}u^{-1}\).
We say \(u \in Σ^*\) is a \emph{characterizing word} for \(q \in Q\) if{}f \(W_{q,F}^{\cN} = u^{-1}\lang{\cN}\) and we say \(\cN\) is \emph{consistent} if{}f every state \(q\) is reachable by a characterizing word for \(q\).
Moreover, \(\cN\) is \emph{strongly consistent} if{}f every state \(q\) is reachable by every characterizing word of \(q\).

Denis et al.~\cite{denis2002residual} define a \emph{residualization} operation that, given NFA \(\cN\), builds an RFA \(\cN^{\text{res}}\) such that \(\lang{\cN^{\text{res}}} = \lang{\cN}\).
Let \(\cN = (Q, Σ, δ, I, F)\) be an NFA and \(u \in Σ^*\), the set \(\post_u^{\cN}(I)\) is \emph{coverable} if{}f \(\post_u^{\cN}(I) = \bigcup_{x\in\Sigma^*, \; \post_x^{\cN}(I) \subsetneq \post_u^{\cN}(I)}\post_{x}^{\cN}(I)\).
Define \(\cN^{\text{res}} \ud (\widetilde{Q}, Σ, \widetilde{δ}, \widetilde{I}, \widetilde{F})\) as an RFA with \(\widetilde{Q} = \{ \post_u^{\cN}(I) \mid u \in Σ^* \land \post_u^{\cN}(I) \text{ is not coverable}\}\), \(\widetilde{I} = \{S \in \widetilde{Q} \mid S \subseteq I\}\), \(\widetilde{F} = \{S \in \widetilde{Q} \mid S \cap F \neq \varnothing\}\) and \(\widetilde{δ}(S, a) = \{S' \in \widetilde{Q} \mid S' \subseteq δ(S, a)\}\) for every \(S \in \widetilde{Q}\) and \(a \in Σ\).

Finally, the \emph{canonical} RFA for a regular language \(L\) is the RFA \(\cC \ud (Q, Σ, δ, I, F)\) with \(Q \!=\! \{u^{-1}L \mid u \!\in\! Σ^* \land u^{-1}L \text{ is prime}\}\), \(I \!=\! \{u^{-1}L \!\in\! Q \mid u^{-1}L \subseteq L\}\), \(F \!=\! \{u^{-1}L \!\in\! Q \mid \varepsilon \!\in\! u^{-1}L\}\) and \(δ(u^{-1}L, a) = \{v^{-1}L \in Q \mid v^{-1}L \subseteq a^{-1}(u^{-1}L)\}\) for every \(u^{-1}L \in Q\) and \(a \in Σ\).
As shown by Denis et al.~\cite{denis2002residual}, the canonical RFA is a strongly consistent RFA and it is the minimal (in number of states) RFA such that \(\lang{\cN} = L\).
Moreover, the canonical RFA is maximal in the number of transitions.

\medskip\noindent\textbf{Quasiorders.}
A \emph{quasiorder} over \(\Sigma^*\) (qo for short) \(\mathord{\qo}\) is a reflexive and transitive binary relation over \(\Sigma^*\).
A \emph{symmetric} qo is called an \emph{equivalence relation}.
A quasiorder \(\preceqq\) is a \emph{right (resp.\ left) quasiorder} and we denote it \(\qr\) (resp. \(\ql\)) if{}f for all \(u,v \in Σ^*\), we have that \(u \preceqq v \Ra ua\preceqq va\) (resp. \(u \preceqq v \Ra au \preceqq av\)), for all \(a \in \Sigma\). 
For example, the quasiorder defined by \(u \qo_{\text{len}} v \udiff \len{u} \leq \len{v}\), is a left and right qo but not an equivalence relation.

Given two qo's \(\mathord{\qo}\) and \(\mathord{\qo'}\), we say that \(\mathord{\qo}\) is \emph{finer} than \(\mathord{\qo'}\) (or \(\mathord{\qo'}\) is \emph{coarser} than \(\mathord{\qo}\)) if{}f \(\mathord{\qo} \subseteq \mathord{\qo'}\).
For every qo \(\qo\), we define its \emph{strict} version as: \(u \qon v \udiff u \qo v \land  v \not\qo u \) and we define \((\qo)^{-1}\) as: \(u~(\qo)^{-1}~v \udiff v \qo u\).
Note that every qo \(\qo\) induces an equivalence relation defined as \(\sim ~\ud~ \qo \cap~ (\qo)^{-1}\).

We adopt the definition of \emph{closure} of a subset of \(S \subseteq Σ^* \) w.r.t. a qo \(\qo\) introduced by de Luca and Varricchio~\cite{deluca2011}.
Concretely, given a qo \({\qo}\) on \(Σ^*\) and a subset \(S \subseteq Σ^*\), we define the \emph{upper closure} (or simply \emph{closure}) of \(S\) w.r.t. \(\qo\) as \(\cl_{\qo}(S) \ud \{w \in Σ^* \mid \exists x \in S, \; x \qo w\}\).
We say that \(\cl_{\qo}(S)\) is a \emph{principal} if{}f \(\cl_{\qo}(S) = \cl_{\qo}(\{u\})\), for some \(u \in \Sigma^*\).
In that case, we write \(\cl_{\qo}(u)\) instead of \(\cl_{\qo}(\{u\})\).
Note that, \(\cl_{\qo}(u) = \cl_{\qo}(v)\), for all \(v \in \Sigma^*\) such that \(u \sim v\).
Finally, given a language \(L \subseteq \Sigma^*\), we say that a qo \(\qo\) is \(L\)-\emph{preserving} if{}f \(\cl_{\qo}(L)=L\).

\section{Automata Constructions from Quasiorders}%
\label{sec:automataConstructions}

We will consider right and left quasiorders on \(\Sigma^*\) (and their corresponding closures) and we will use them to define RFAs constructions for regular languages.
The following lemma gives a characterization of right and left quasiorders.

\begin{restatable}{lemmaR}{CongruencePbwComplete}%
\label{lemma:QObwComplete}
The following properties hold:
\begin{enumerate}
\item \(\qr\) is a right quasiorder if{}f \(\cl_{\qr}(u) v \subseteq \cl_{\qr}(uv)\), for all \(u,v \in \Sigma^*\).
\item\(\ql\) is a left quasiorder if{}f \(v \cl_{\ql}(u) \subseteq \cl_{\ql}(vu)\), for all \(u,v \in \Sigma^*\).
\end{enumerate}

\end{restatable}

Given a regular language \(L\), we are interested in left and right quasiorders that are \(L\)-preserving.
We will use the principals of these quasiorders as states of automata constructions that yield RFAs and co-RFAs accepting the language \(L\).
Therefore, in the sequel, we will only consider quasiorders that induce a finite number of principals, i.e., quasiorders \(\qo\) such that the induced equivalence \(\mathord{\sim} \ud \mathord{\qo} \cap (\mathord{\qo})^{-1}\) has finite index.

Next, we introduce the notion of \emph{\(L\)-composite principals} which, intuitively, correspond to states of our automata constructions that can be removed without altering the language accepted by the automata.

\begin{definition}[\(L\)-Composite Principal]%
\label{def:CompositeClosed}
Let \(L\) be a regular language and let \(\qr\) (resp. \(\ql\)) be a right (resp.\ left) quasiorder on \(Σ^*\).
Given \(u \in \Sigma^*\), the principal \(\cl_{\qr}(u)\) (resp. \(\cl_{\ql}(u)\)) is \(L\)-\emph{composite} if{}f %
\begin{align*}
u^{-1}L & = \hspace{-10pt}\bigcup_{x\in\Sigma^*,\; x \qrn u }\hspace{-10pt} x^{-1}L &
\text{\emph{(}resp. }Lu^{-1} & = \hspace{-10pt}\bigcup_{x\in\Sigma^*,\; x \qln u }\hspace{-10pt} Lx^{-1}\text{\emph{)}}
\end{align*}
If \(\cl_{\qr}(u)\) (resp. \(\cl_{\ql}(u)\)) is not \(L\)-composite then it is \emph{\(L\)-prime}.
\end{definition}

We sometimes use the terms \emph{composite} and \emph{prime principal} when the language \(L\) is clear from the context.
Observe that, if \(\cl_{\qr}(u)\) is \(L\)-composite, for some \(u \in \Sigma^*\), then so is \(\cl_{\qr}(v)\), for every \(v \in \Sigma^*\) such that \(u \rr v\).
The same holds for a left quasiorder \(\ql\).

Given a regular language \(L\) and a right quasiorder \(\qr\) that is \(L\)-preserving, the following automata construction yields an RFA that accepts exactly the language \(L\).

\begin{definition}[Automata construction \(\cH^{r}(\qr, L)\)]%
\label{def:right-const:qo}
Let \(\qr\) be a right quasiorder and let \(L \subseteq \Sigma^*\) be a language.
Define the automaton \(\cH^{r}(\qr, L) ≝ (Q, \Sigma, \delta, I, F)\) where \(Q = \{\cl_{\qr}(u) \mid u \in Σ^*, \; \cl_{\qr}(u) \text{ is \(L\)-prime}\}\), \(I = \{\cl_{\qr}(u) \in Q \mid \varepsilon \in \cl_{\qr}(u)\}\), \(F = \{\cl_{\qr}(u) \in Q \mid u \in L\}\) and \( \delta(\cl_{\qr}(u), a) = \{ \cl_{\qr}(v) \in Q \mid \cl_{\qr}(u)   a \subseteq \cl_{\qr}(v)\}\) for all \(\cl_{\qr}(u) \in Q, a \in Σ\).
\end{definition}

\begin{restatable}{lemmaR}{HrGeneratesL}%
\label{lemma: HrGeneratesL}
Let \(L\subseteq \Sigma^*\) be a regular language and let \(\qr\) be a right \(L\)-preserving quasiorder.
Then \(\cH^r(\qr,L)\) is an RFA such that \(\lang{\cH^r(\qr,L)} = L\).
\end{restatable}

Given a regular language \(L\) and a left \(L\)-preserving quasiorder \(\ql\), we can give a similar automata construction of a co-RFA that recognizes exactly the language \(L\).

\begin{definition}[Automata construction \(\cH^{\ell}(\ql, L)\)]%
\label{def:left-const:qo}
Let \(\ql\) be a left quasiorder and let \(L \subseteq \Sigma^*\) be a language.
Define the automaton \(\cH^{\ell}(\ql, L)= (Q, \Sigma, \delta, I, F)\) where \(Q = \{\cl_{\ql}(u) \mid \linebreak u \in Σ^*, \; \cl_{\ql}(u) \text{ is \(L\)-prime}\}\), \(I = \{ \cl_{\ql}(u) \in Q \mid u \in L \}\), \(F = \{\cl_{\ql}(u) \in Q \mid \varepsilon \in \cl_{\ql}(u)\}\), and \(\delta(\cl_{\ql}(u), a) = \{\cl_{\ql}(v)\in Q \mid a   \cl_{\ql}(v) \subseteq \cl_{\ql}(u)\}\) for all \(\cl_{\ql}(u) \in Q, a \in \Sigma\).
\end{definition}

\begin{restatable}{lemmaR}{HlgeneratesL}%
\label{lemma:HlgeneratesL}
Let \(L\subseteq \Sigma^*\) be a language and let \(\ql\) be a left \(L\)-preserving quasiorder.
Then \(\cH^{\ell}(\ql,L)\) is a co-RFA such that \(\lang{\cH^{\ell}(\ql,L)} = L\).
\end{restatable}

Observe that the automaton \(\cH^r = \cH^{r}(\qr, L)\) (resp. \(\cH^{\ell} = \cH^{\ell}(\ql, L)\)) is \emph{finite}, since we assume \(\qr\) (resp. \(\ql\)) induces a finite number of principals.
Note also that \(\cH^{r}\) (resp. \(\cH^{\ell}\)) possibly contains empty (resp.\ unreachable) states but no state is unreachable (resp.\ empty).

Moreover, notice that by keeping all principals of \(\qr\) (resp. \(\ql\)) as states, instead of only the prime ones as in Definition~\ref{def:right-const:qo} (resp. Definition~\ref{def:left-const:qo}), we would obtain an RFA (resp.\ a co-RFA) with (possibly) more states that also recognizes \(L\).

The following lemma shows that \(\cH^r\) and \(\cH^{\ell}\) inherit the left-right duality between \(\qr\) and \(\ql\) through the reverse operation.

\begin{restatable}{lemmaR}{leftRightReverse}%
\label{lemma:leftRightReverse}
Let \(\qr\) and \(\ql\) be a right and a left quasiorder, respectively, and let \(L \subseteq \Sigma^*\) be a language.
If \(u \qr v \Lra u^R \ql v^R\) then \(\cH^{r}(\qr, L) \) is isomorphic to \( \left(\cH^{\ell}(\ql, L^R)\right)^R\).
\end{restatable}

Finally, it follows from the next theorem that given two right \(L\)-preserving quasiorders, \(\qr_1\) and \(\qr_2\), if \(\mathord{\qr_1} \subseteq \mathord{\qr_2}\) then the automaton \(\cH^{r}(\qr_1,L)\) has, at least, as many states as \(\cH^{r}(\qr_2,L)\).
The same holds for left \(L\)-preserving quasiorders and \(\cH^{\ell}\).
Observe that this is not obvious since only the \(L\)-prime principals correspond to states of the automata construction.

\begin{restatable}{theoremR}{numLPrimePrincipals}\label{theorem:numLPrimePrincipals}
Let \(L\subseteq Σ^*\) be a language and let \(\qo_1\) and \(\qo_2\) be two left or two right \(L\)-preserving quasiorders.
If \(\mathord{\qo_1} \subseteq \mathord{\qo_2}\) then:
\[\len{\{\cl_{\qo_1}(u) \mid u \in Σ^* \land \cl_{\qo_1}(u) \text{ is \(L\)-prime}\}} \geq \len{\{\cl_{\qo_2}(u) \mid u \in Σ^* \land \cl_{\qo_2}(u) \text{ is \(L\)-prime}\}} \enspace .\]
\end{restatable}

\section{Language-based Quasiorders and their Approximation using NFAs}%
\label{sec:Instantiation}
In this section we instantiate our automata constructions using two classes of quasiorders, namely, the so-called \emph{Nerode's} quasiorders~\cite{de1994well}, whose definition is based on a given regular language; and the \emph{automata-based} quasiorders, whose definition relies on a given automaton.

\begin{definition}[Language-based Quasiorders]%
\label{def:NerodeQO}
Let \(u,v \in \Sigma^*\) and let \(L \subseteq \Sigma^*\) be a language. 
Define:
\begin{align}
u \qrL v  & \udiff u^{-1}L \subseteq v^{-1}L & \quad \text{\emph{Right-}language-based Quasiorder}\label{eq:Rlanguage} \\
u \qlL v  & \udiff Lu^{-1} \subseteq Lv^{-1}  & \quad \text{\emph{Left-}language-based Quasiorder}%
\label{eq:Llanguage}
\end{align}
\end{definition}

It is well-known that for every regular language \(L\) there exists a finite number of quotients \(u^{-1}L\)~\cite{deluca2011}	.
Therefore, the language-based quasiorders defined above induce a finite number of principals since each principal set is determined by a quotient of \(L\).

\begin{definition}[Automata-based Quasiorders]%
\label{def:automataQO}
Let \(u,v \in\Sigma^*\) and let \(\cN = (Q, \Sigma, \delta, I, F)\) be an NFA\@.
Define:
\begin{align}
u \qrN v & \udiff \post^{\cN}_u(I) \subseteq \post^{\cN}_v(I) & \quad \text{\emph{Right-}Automata-based Quasiorder}\label{eq:RState} \\
u \qlN v & \udiff \pre^{\cN}_u(F) \subseteq \pre^{\cN}_v(F) & \quad \text{\emph{Left-}Automata-based Quasiorder} \label{eq:LState} 
\end{align}
\end{definition}

Clearly, the automata-based quasiorders induce a finite number of principals since each principal is represented by a subset of the states of \(\cN\).

\begin{remark} \label{remark:LRDual}
The pairs of quasiorders \(\qrL\)~-~\(\qlL\) and \(\qrN\)~-~\(\qlN\) from Definitions~\ref{def:NerodeQO} and~\ref{def:automataQO} are dual, i.e.  \(u \qrL v \Lra u^R \qlL v^R\) and \(u \qrN v \Lra u^R \qlN v^R\).
\end{remark}

The following result shows that the principals of \(\qrN\) and \(\qlN\) can be described, respectively, as intersections of left and right languages of the states of \(\cN\) while the principals of \(\qrL\) and \(\qlL\) correspond to intersections of quotients of \(L\).

\begin{restatable}{lemmaR}{positiveAtoms}%
\label{lemma:positive_atoms}
Let \(\cN = (Q,Σ,δ,I,F)\) be an NFA with \(\lang{\cN}=L\).  
Then, for every \(u \in \Sigma^*\),
\begin{align*}
\cl_{\qrN}(u) &= \bigcap\textstyle{_{q \in \post_u^{\cN}(I)}} W_{I,q}^{\cN} & \cl_{\qrL}(u)  &= \bigcap\textstyle{_{w \in \Sigma^*, \; w \in u^{-1}L}} Lw^{-1} \\
\cl_{\qlN}(u)  & = \bigcap\textstyle{_{q \in \pre_u^{\cN}(I) }} W_{q,F}^{\cN} & \cl_{\qlL}(u)  &= \bigcap\textstyle{_{w \in \Sigma^*, \; w \in Lu^{-1}}} w^{-1}L\enspace .
\end{align*}
\end{restatable}
\smallskip

As shown by Ganty et al.~\cite{ganty2019Inclusion}, given an NFA \(\cN\) with \(L = \lang{\cN}\), the quasiorders \(\qrL\) and \(\qrN\) are right \(L\)-preserving quasiorders, while the quasiorders \(\qlL\) and \(\qlN\) are left \(L\)-preserving quasiorders.
Therefore, by Lemma~\ref{lemma: HrGeneratesL} and~\ref{lemma:HlgeneratesL}, our automata constructions applied to these quasiorders yield automata for \(L\).

Finally, as shown by de Luca and Varricchio~\cite{de1994well}, we have that \(\qrN\) is finer than \(\qrL\), i.e., \(\mathord{\qrN} \subseteq \mathord{\qrL}\).
In that sense we say that \(\qrN\) \emph{approximates}  \(\qrL\).
As the following lemma shows, the approximation is precise, i.e., \(\qrN ~=~ \qrL\), whenever \(\cN\) is a co-RFA with no empty states\@.

\begin{restatable}{lemmaR}{coResidualqrLqrN}\label{lemma:coResidual_qrL=qrN}
Let \(\cN\) be a co-RFA with no empty states such that \(L = \lang{\cN}\).
Then \(\mathord{\qrL} = \mathord{\qrN}\).
Similarly, if \(\cN\) is an RFA with no unreachable states and \(L = \lang{\cN}\) then \(\mathord{\qlL} = \mathord{\qlN}\).
\end{restatable}

\subsection{Automata Constructions}

In what follows, we will use \(\cF{r},\cF{\ell}\) and \(\cG{r},\cG{\ell}\) to denote the constructions \(\cH^r, \cH^{\ell}\) when applied, respectively, to the language-based quasiorders induced by a regular language and the automata-based quasiorders induced by an NFA\@.

\begin{definition}%
\label{def:FG}
Let \(\cN\) be an NFA accepting the language \(L = \lang{\cN}\).
Define:
\begin{align*}
\cF{r}(L) & ≝  \cH^{r}(\qrL, L) & \cG{r}(\cN) & ≝ \cH^{r}(\qrN, L) \\
\cF{\ell}(L) & ≝  \cH^{\ell}(\qlL, L) & \cG{\ell}(\cN) & ≝  \cH^{\ell}(\qlN, L) \enspace .
\end{align*}
\end{definition}
\smallskip

Given an NFA \(\cN\) accepting the language \(L=\lang{\cN}\), all constructions in the above definition yield automata accepting \(L\).
However, while the constructions using the right quasiorders result in RFAs, those using left quasiorders result in co-RFAs.
Furthermore, it follows from Remark~\ref{remark:LRDual} and Lemma~\ref{lemma:leftRightReverse} that \(\cF{\ell}(L)\) is isomorphic to \((\cF{r}(L^R))^R\) and \(\cG{\ell}(\cN)\) is isomorphic to \((\cG{r}(\cN^R))^R\).

It follows from Theorem~\ref{theorem:numLPrimePrincipals} that the automata \(\cG{r}(\cN)\) and \(\cG{\ell}(\cN)\) have more states than \(\cF{r}(L)\) and \(\cF{\ell}(L)\), respectively.
Intuitively, \(\cF{r}(L)\) is the minimal RFA for \(L\), i.e. it is isomorphic to the canonical RFA for \(L\), since
\(\qrL\) is the coarsest right \(L\)-preserving quasiorder~\cite{de1994well}.
On the other hand, as we evidenced in Example~\ref{example:residualization}, \(\cG{r}(\cN)\) is a sub-automaton of \(\cN^{\text{res}}\)~\cite{denis2002residual} for every NFA \(\cN\).

Finally, it follows from Lemma~\ref{lemma:coResidual_qrL=qrN} that residualizing (\(\cG{r}\)) a co-RFA with no empty states (\(\cG{\ell}(\cN)\)) results in the canonical RFA for \(\lang{\cN}\) (\(\cF{r}(\lang{\cN})\)).

We formalize all these notions in Theorem~\ref{theoremF}.
Figure~\ref{Figure:diagramAutomata} summarizes all these connections between the automata constructions given in Definition~\ref{def:FG}.

\begin{figure}[t]
\begin{minipage}[l]{0.45\textwidth}
\begin{tikzcd}[column sep=small, row sep=normal]
\cN\ar[d, description, "R"',leftrightarrow] \ar[r, "\cG{\ell}"] \ar[rr, start anchor=70, bend left=20, "\cF{r}"] & \cG{\ell}(\cN) \ar[d, description, "R"',leftrightarrow] \ar[r, "\cG{r}"] & \cF{r}(\lang{\cN}) \ar[d, description, "R"',leftrightarrow]\\
\cN^R \ar[r, "\cG{r}"] \ar[rr, start anchor=290, bend right=20, "\cF{\ell}"] & \cG{r}(\cN^R) \ar[r, "\cG{\ell}"] & \cG{\ell}(\cG{r}(\cN^R))
\end{tikzcd}
\end{minipage}\hfill
\begin{minipage}[r]{0.54\textwidth}
The upper part of the diagram follows from Theorem~\ref{theoremF}~(\ref{lemma:LS+RS=RN}), the squares follow from Theorem~\ref{theoremF}~(\ref{lemma:AlRequalArNR}) and the bottom curved arc follows from Theorem~\ref{theoremF}~(\ref{lemma:FlisomorphicRfrR}).
Incidentally, the diagram shows a new relation which is a consequence of the left-right dualities between \(\qlL\) and \(\qrL\), and \(\qlN\) and \(\qrN\): \(\cF{\ell}(\lang{\cN^R})\) is isomorphic to \(\cG{\ell}(\cG{r}(\cN^R))\).
\end{minipage}
\caption{Relations between the constructions \(\cG{\ell},\cG{r},\cF{\ell}\) and \(\cF{r}\).
Note that constructions \(\cF{r}\) and \(\cF{\ell}\) are applied  to the language accepted by the automaton in the origin of the labeled arrow while constructions \(\cG{r}\) and \(\cG{\ell}\) are applied directly to the automaton.}\label{Figure:diagramAutomata}
\end{figure}

\begin{restatable}{theoremR}{theoremF}\label{theoremF}
Let \(\cN\) be an NFA with \(L = \lang{\cN}\).
Then the following properties hold:
\begin{alphaenumerate}%
\item \(\lang{\cF{r}(L)} =\lang{\cF{\ell}(L)} = L = \lang{\cG{r}(\cN)} = \lang{\cG{\ell}(\cN)}\).%
\label{lemma:language-F}
\item \(\cF{\ell}(L)\) is isomorphic to \((\cF{r}(L^R))^R\).%
\label{lemma:FlisomorphicRfrR}
\item \(\cG{\ell}(\cN)\) is isomorphic to \((\cG{r}(\cN^R))^R\).%
\label{lemma:AlRequalArNR}
\item \(\cF{r}(L)\) is isomorphic to the canonical RFA for \(L\).%
\label{theorem:CanonicalRFAlanguage}
\item \(\cG{r}(\cN)\) is isomorphic to a sub-automaton of \(\cN^{\text{res}}\) and \(\lang{\cG{r}(\cN)} = \lang{\cN^{\text{res}}} = L\).%
\label{lemma:rightNRes}
\item \(\cG{r}(\cG{\ell}(\cN))\) is isomorphic to \(\cF{r}(L)\).\label{lemma:LS+RS=RN}
\end{alphaenumerate}
\end{restatable}
\medskip

Let \(\cN\) be an NFA with \(L = \lang{\cN}\).
If \(\mathord{\qrL}=\mathord{\qrN}\) then the automata \(\cF{r}(L)\) and \(\cG{r}(\cN)\) are isomorphic.
The following result shows that the reverse implication also holds.

\begin{restatable}{lemmaR}{qrlqrNCanRes}\label{lemma:qrlEqualqrNResEqualCan}
Let \(\cN\) be an NFA with \(L = \lang{\cN}\)\@.
Then \(\mathord{\qrL} = \mathord{\qrN}\) if{}f \(\cG{r}(\cN)\) is isomorphic to \(\cF{r}(\lang{\cN})\).
\end{restatable}
\medskip

The following example illustrates the differences between our residualization operation, \(\cG{r}(\cN)\), and the one defined by Denis et al.~\cite{denis2002residual}, \(\cN^{\text{res}}\), on a given NFA \(\cN\): the automaton \(\cG{r}(\cN)\) has, at most, as many states as \(\cN^{\text{res}}\).
This follows from the fact that for every \(u \in Σ^*\), if \(\post_u^{\cN}(I)\) is coverable then \(\cl_{\qrN}(u)\) is composite but \emph{not} vice-versa.

\begin{example}%
\label{example:residualization}
Let \(\cN = (Q, Σ, δ, I, F)\) be the automata on the left of Figure~\ref{fig:Residuals} and let \(L = \lang{\cN}\).
To build \(\cN^{\text{res}}\) we compute \(\post_u^{\cN}(I)\), for all \(u \in Σ^*\).
Let \(C \ud L^c \setminus \{\varepsilon, a, b, c\}\).
\begin{align*}
\post_{\varepsilon}^{\cN}(I)  & = \{0\} & \post_a^{\cN}(I)  & = \{1,2\}& \forall w \in L, \;\post_{w}^{\cN}(I)  & = \{5\} \\
\post_c^{\cN}(I)  & = \{1, 2, 3, 4\} & \post_b^{\cN}(I)  & = \{1,3\} &\forall w \in C,\; \post_{w}^{\cN}(I) & = \varnothing
\end{align*}
Since none of these sets is coverable by the others, they are all states of \(\cN^{\text{res}}\).
The resulting RFA \(\cN^{\text{res}}\) is shown in the center of Figure~\ref{fig:Residuals}.
On the other hand, let us denote \(\cl_{\qrN}\) simply by \(\cl\).
In order to build \(\cG{r}(\cN)\) we need to compute the principals \(\cl(u)\), for all \(u \in Σ^*\).
By definition of \(\qrN\), we have that \(w \in \cl(u) \Lra \post_u^{\cN}(I) \subseteq \post_w^{\cN}(I)\).
Therefore, we obtain:
\[
\cl(\varepsilon)  {=} \{\varepsilon\}  \;\; \cl(a)  {=} \{a,c\} \;\; \cl(b)  {=} \{b,c\} \;\;
\cl(c)  {=} \{c\} \quad \forall w \in L,\;  \cl(w)  {=} L \quad \forall w \in C,\;\cl(w) {=} \Sigma^* \enspace .\]
Since \(a \qrn_{\cN} c\), \(b \qrn_{\cN} c\) and \(\forall w\in Σ^*, \; cw \subseteq L \Lra \big(aw \subseteq L \lor bw \subseteq L\big)\), it follows that \(\cl(c)\) is \(L\)-composite.
The resulting RFA \(\cG{r}(\cN)\) is shown on the right of Figure~\ref{fig:Residuals}.
\eox%
\end{example}

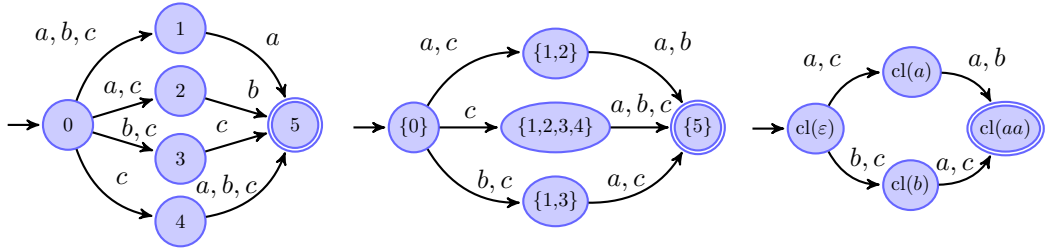
\begin{figure}[t]%
    \centering
\begin{minipage}[l]{0.32\textwidth}
  \begin{tikzpicture}[->,>=stealth',shorten >=1pt,auto,node distance=3mm and 8mm,thick,initial text=]
  \tikzstyle{every state}=[scale=0.75,fill=blue!20,draw=blue!60,text=black]

  \node[initial, state] (0) {\(0\)};
  \node[state] (2) [right=of 0, yshift=0.6cm] {\(2\)};
  \node[state] (1) [above=of 2, yshift=-0.2cm] {\(1\)};
  \node[state] (3) [right=of 0, yshift=-0.6cm]{\(3\)};
  \node[state] (4) [below=of 3, yshift=0.2cm] {\(4\)};
  \node[accepting, state] (5) [right=of 0, xshift=2cm] {\(5\)};
  
  \path (0) edge [bend left] node {\(a,b,c\)} (1)
  		(0) edge node [xshift=10pt, yshift=0pt] {\(a,c\)} (2)
  		(0) edge node [pos=0.3, yshift=-5pt] {\(b,c\)} (3)
  		(0) edge [bend right] node {\(c\)} (4);

  \path (1) edge [bend left] node {\(a\)} (5)
  		(2) edge node {\(b\)} (5)
  		(3) edge node {\(c\)} (5)
  		(4) edge [bend right] node [yshift=-5pt, xshift=4pt] {\(a,b,c\)} (5);
  \end{tikzpicture}
  \end{minipage}\hfill%
  \begin{minipage}{0.37\textwidth}
  \begin{tikzpicture}[->,>=stealth',shorten >=1pt,auto,node distance=3mm and 8mm,thick,initial text=]
  \tikzstyle{every state}=[scale=0.75,fill=blue!20,draw=blue!60,text=black,style={draw,ellipse,inner sep=0pt}]

  \node[initial, state] (0) {\(\left\{0\right\}\)};
  \node[state] (2) [right=of 0] {\(\left\{1{,}2{,}3{,}4\right\}\)};
  \node[state] (1) [above=of 2] {\(\left\{1{,}2\right\}\)};
  \node[state] (3) [below=of 2] {\(\left\{1{,} 3\right\}\)};
  \node[accepting, state] (4) [right=of 2] {\(\left\{5\right\}\)};
  
  \path (0) edge [bend left] node {\(a,c\)} (1)
  		(0) edge node {\(c\)} (2)
  		(0) edge [bend right] node [yshift=-5pt] {\(b,c\)} (3);

  \path (1) edge [bend left] node {\(a,b\)} (4)
  		(2) edge node {\(a,b,c\)} (4)
  		(3) edge [bend right] node [xshift=4pt, yshift=-4pt] {\(a,c\)} (4);
  \end{tikzpicture}
  \end{minipage}\hfill%
  \begin{minipage}[r]{0.3\textwidth}
  \begin{tikzpicture}[->,>=stealth',shorten >=1pt,auto,node distance=3mm and 5mm,thick,initial text=]
  \tikzstyle{every state}=[scale=0.75,fill=blue!20,draw=blue!60,text=black,style={draw,ellipse,inner sep=0pt}]
  
  \node[initial, state] (0) {\(\cl(\varepsilon)\)};
  \node[state] (1) [right=of 0, yshift=1cm] {\(\cl(a)\)};
  \node[state] (2) [right=of 0, yshift=-1cm] {\(\cl(b)\)};
  \node[accepting, state] (3) [right=of 0, xshift=1.5cm] {\(\cl(aa)\)};
  
  \path (0) edge [bend left] node {\(a,c\)} (1)
  		(0) edge [bend right] node [yshift=-2pt, xshift=-4pt] {\(b,c\)} (2);

  \path (1) edge [bend left] node [xshift=-5pt]{\(a,b\)} (3)
  		(2) edge [bend right] node [xshift=5pt, yshift=-3pt] {\(a,c\)} (3);
  \end{tikzpicture}
  \end{minipage}
\caption{An NFA \(\cN\) and the RFAs \(\cN^{\text{res}}\) and \(\cG{r}(\cN)\). We omit the empty states for clarity.}
 \label{fig:Residuals}
\end{figure}

\section{Double-Reversal Method for Building the Canonical RFA}%
\label{sec:Novel}

Denis et al.~\cite{denis2002residual} show that their residualization operation satisfies the residual-equivalent of the double-reversal method for building the minimal DFA\@. 
More specifically, they prove that if an NFA \(\cN\) is a co-RFA with no empty states, then their residualization operation applied to \(\cN\) results in the canonical RFA for \(\lang{\cN}\).
As a consequence, \((((\cN^R)^{\text{res}})^R)^{\text{res}}\) is the canonical RFA for \(\lang{\cN}\). 

In this section we first show that the residual-equivalent of the double-reversal method holds within our framework, i.e.\ \(\cG{r}((\cG{r}(\cN^R))^R)\) is isomorphic to \(\cF{r}(\cN)\).
Then, we generalize this method along the lines of the generalization of the double-reversal method for building the minimal DFA given by Brzozowski and Tamm~\cite{brzozowski2014theory}.
To this end, we extend our previous work~\cite{ganty2019congruence} in which we provided a congruence-based perspective on the generalized double-reversal method for DFAs.
By moving from congruences to quasiorders, we find a \emph{necessary} and \emph{sufficient} condition on an NFA \(\cN\) so that \(\cG{r}(\cN)\) yields the canonical RFA for \(\lang{\cN}\).
Finally, we compare our generalization with the one given by Tamm~\cite{tamm2015generalization}.

\subsection{Double-reversal Method}
We give a simple proof of the double-reversal method for building the canonical RFA\@.

\begin{theorem}[Double-Reversal]%
\label{theorem:DoubleReversal}
Let \(\cN\) be an NFA\@.
Then \(\cG{r}((\cG{r}(\cN^R))^R)\) is isomorphic to the canonical RFA for \(\lang{\cN}\).
\end{theorem}
\begin{proof}
It follows from Theorem~\ref{theoremF}~(\ref{lemma:AlRequalArNR}), (\ref{theorem:CanonicalRFAlanguage}) and (\ref{lemma:LS+RS=RN}).
\end{proof} 
Note that Theorem~\ref{theorem:DoubleReversal} can be inferred from Figure~\ref{Figure:diagramAutomata} by following the path starting at \(\cN\), labeled with \(R-\cG{r}-R-\cG{r}\) and ending in \(\cF{r}(\lang{\cN})\).

\subsection{Generalization of the Double-reversal Method}
Next we show that residualizing an automaton yields the canonical RFA if{}f the left language of every state is closed w.r.t. the right Nerode quasiorder.

\begin{theorem}%
\label{theorem:canonicalreverserestic}
Let \(\cN = (Q,\Sigma,\delta,I,F)\) be an NFA with \(L=\lang{\cN}\).
Then \(\cG{r}(\cN)\) is the canonical RFA for \(L\) if{}f \(\forall q \in Q,\;  \cl_{\qrL}(W_{I,q}^{\cN}) = W_{I,q}^{\cN}\).
\end{theorem}
\begin{proof}
We first show that \(\forall q \in Q,\ \cl_{\qrL}(W_{I,q}^{\cN}) = W_{I,q}^{\cN}\) is a \emph{necessary} condition, i.e.\ if \(\cG{r}(\cN)\) is the canonical RFA for \(L\) then \(\forall q \in Q,\ \cl_{\qrL}(W_{I,q}^{\cN}) = W_{I,q}^{\cN}\) holds.
By Lemma~\ref{lemma:qrlEqualqrNResEqualCan} we have that if \(\cG{r}(\cN)\) is the canonical RFA then \(\mathord{\qrL} = \mathord{\qrN}\).
Moreover,
\begin{align*}
\cl_{\qrL}(W_{I,q}^{\cN}) & = \quad \text{[By definition of \(\cl_{\qrL}\)]}\\
\{w \in \Sigma^* \mid \exists u \in W_{I,q}^{\cN}, \; u^{-1}L \subseteq w^{-1}L\} & = \quad \text{[Since \(\mathord{\qrL} = \mathord{\qrN}\)]} \\
\{w \in \Sigma^* \mid \exists u \in W_{I,q}^{\cN}, \; \post_u^{\cN}(I) \subseteq \post_w^{\cN}(I)\} & \subseteq  \quad \text{[Since \( u \in W_{I,q}^{\cN} \Lra q \in \post_u^{\cN}(I)\)]}\\
\{w \in \Sigma^* \mid q \in \post_w^{\cN}(I)\} & = \quad \text{[By definition of \(W_{I,q}^{\cN}\)]}\\
W_{I,q}^{\cN} \enspace .
\end{align*}
By reflexivity of \(\qrL, \) we conclude that \(\cl_{\qrL}(W_{I,q}^{\cN}) = W_{I,q}^{\cN} \).

Next, we show that \(\forall q \in Q,\;  \cl_{\qrL}(W_{I,q}^{\cN}) = W_{I,q}^{\cN}\) is also a \emph{sufficient} condition.
By Lemma~\ref{lemma:positive_atoms} and condition \(\forall q \in Q,\;  \cl_{\qrL}(W_{I,q}^{\cN}) = W_{I,q}^{\cN}\), we have that
\[\cl_{\qrN}(u) = \bigcap{\textstyle{_{q \in \post^{\cN}_u(I)}}} W_{I,q}^{\cN} = \bigcap{\textstyle{_{q \in \post^{\cN}_u(I)}}} \cl_{\qrL}(W_{I,q}^{\cN})\enspace .\]

Since \(u \in \cl_{\qrL}(W_{I,q}^{\cN})\) for all \(q \in \post_u^{\cN}(I)\), it follows that \(\cl_{\qrL}(u) \subseteq \cl_{\qrL}(W_{I,q}^{\cN})\) for all \(q \in \post_u^{\cN}(I)\) and, since \(\cl_{\qrN}(u) = \bigcap\textstyle{_{q \in \post^{\cN}_u(I)}} \cl_{\qrL}(W_{I,q}^{\cN})\), we have that \(\cl_{\qrL}(u) \subseteq \cl_{\qrN}(u)\) for every \(u \in Σ^*\), i.e., \(\mathord{\qrL} \subseteq \mathord{\qrN}\).

On the other hand, as shown by de Luca and Varricchio~\cite{de1994well}, we have that \(\mathord{\qrN}\subseteq \mathord{\qrL}\).
We conclude that \(\mathord{\qrN} = \mathord{\qrL}\), hence \(\cG{r}(\cN) = \cF{r}(L)\).
\end{proof}

It is worth to remark that Theorem~\ref{theorem:canonicalreverserestic} does not hold when considering the residualization operation \(\cN^{\text{res}}\) of Denis et al.~\cite{denis2002residual} instead of \(\cG{r}(\cN)\).
As a counterexample we have the automata \(\cN\) in Figure~\ref{fig:Residuals} where  \(\cG{r}(\cN)\) is the canonical RFA for \(\lang{\cN}\), hence \(\cN\) satisfies the condition of Theorem~\ref{theorem:canonicalreverserestic}, while \(\cN^{\text{res}}\) is not canonical.

\paragraph*{Co-atoms and co-rests}
The condition of Theorem~\ref{theorem:canonicalreverserestic} is analogue to the one we gave for building the minimal DFA~\cite{ganty2019congruence}, except that the later is formulated in terms of congruences instead of quasiorders.
In that case we proved that determinizing a given NFA \(\cN\) yields the minimal DFA if{}f \(\cl_{\rrL}(W^{\cN}_{I,q}) = W^{\cN}_{I,q}\) for every state \(q\) of \(\cN\), where \(\mathord{\rrL} \ud \mathord{\qrL} \cap \mathord{(\qrL)^{-1}}\) is the right Nerode's congruence~\cite{deluca2011}.

Moreover, we showed that the principals of \(\rrL\) coincide with the so-called \emph{co-atoms}~\cite{ganty2019congruence}, which are non-empty intersections of complemented and uncomplemented right quotients of the language.
This allowed us to connect our result for DFAs~\cite{ganty2019congruence} with the generalization of the double-reversal method for building the minimal DFA proposed by Brzozowski and Tamm~\cite{brzozowski2014theory}, who establish that determinizing an NFA \(\cN\) yields the minimal DFA for \(\lang{\cN}\) if{}f the left languages of the states of \(\cN\) are unions of co-atoms of \(\lang{\cN}\).

Next, we give a formulation of the condition from Theorem~\ref{theorem:canonicalreverserestic} along the lines of the one given by Brzozowski and Tamm~\cite{brzozowski2014theory} for their generalization of the double-reversal method for building the minimal DFA.
To do that, let us call the intersections used in Lemma~\ref{lemma:positive_atoms} to describe the principals of \(\qlL\) and \(\qrL\) as \emph{rests} and \emph{co-rests} of \(L\), respectively.
As shown by Theorem~\ref{theorem:canonicalreverserestic}, residualizing an NFA \(\cN\) yields the canonical RFA for \(\lang{\cN}\) if{}f the left language of every state of \(\cN\) satisfies \(\cl_{\qrL}(W_{I,q}^{\cN}) = W_{I,q}^{\cN}\).
By definition, \(\cl_{\qrL}(S) = S\) if{}f \(S\) is a union of principals of \(\qrL\) which, by Lemma~\ref{lemma:positive_atoms} are the co-rests of \(L\).

Therefore we derive the following statement, equivalent to Theorem~\ref{theorem:canonicalreverserestic}, that we consider as the residual-equivalent of the generalization of the double-reversal method for building the minimal DFA proposed by Brzozowski and Tamm~\cite{brzozowski2014theory}.

\begin{corollary}%
Let \(\cN = (Q,\Sigma,\delta,I,F)\) be an NFA with \(L=\lang{\cN}\).
Then \(\cG{r}(\cN)\) is the canonical RFA for \(L\) if{}f the left languages of \(\cN\) are union of co-rests.
\end{corollary}

\paragraph*{Tamm's Generalization of the Double-reversal Method for RFAs}
Tamm~\cite{tamm2015generalization} generalized the double-reversal method of Denis et al.~\cite{denis2002residual} by showing that \(\cN^{\text{res}}\) is the canonical RFA for \(\lang{\cN}\) if{}f the left languages of \(\cN\) are union of the left languages of the canonical RFA for \(\lang{\cN}\).

In this section, we compare the generalization of Tamm~\cite{tamm2015generalization} with ours.
The two approaches differ in the definition of the residualization operation they consider and, as the following lemma shows, the sufficient and necessary condition from Theorem~\ref{theorem:canonicalreverserestic} is more general than that of Tamm~\cite[Theorem 4]{tamm2015generalization}

\begin{lemma}\label{lemma:WeRaTamm}
Let \(\cN = (Q, Σ, δ, I, F)\) be an NFA and let \(\cC = \cF{r}(\qrL, L) = (\wt{Q}, Σ, \wt{δ}, \wt{I}, \wt{F})\) be the canonical RFA for \(L =\lang{\cN}\).
If \(W_{I,q}^{\cN} = \bigcup_{q \in \wt{Q}}W_{\wt{I},q}^{\cC}\) then \( \cl_{\qrL}(W_{I,q}^{\cN}) = W_{I,q}^{\cN}\).
\end{lemma}

\begin{proof}
Since the canonical RFA, \(\cC\), is strongly consistent, it follows from Lemma~\ref{lemma:qrHEqualqrifHsc} (see Appendix~\ref{sec:supp-results}) that \(\mathord{\qr_{\mathcal{C}}} = \mathord{\qrL}\) and, consequently, \(\cG{r}(\cC)\) is isomorphic to \(\cF{r}(L)\).
It follows from Theorem~\ref{theorem:canonicalreverserestic} that \(\cl_{\qrL}(W_{\wt{I},q}^{\cC}) = W_{\wt{I},q}^{\cC}\) for every \(q \in \wt{Q}\).
Therefore,
\begin{align*}
\cl_{\qrL}(W_{I,q}^{\cN}) & = \quad \text{[Since \(W_{I,q}^{\cN} = {\textstyle\bigcup_{q \in \wt{Q}}W_{\wt{I},q}^{\cC}}\) and \(\cl_{\qrL}(\cup S_i) = \cup \cl_{\qrL}(S_i)\)]} \\
{\textstyle\bigcup_{q \in \wt{Q}}\cl_{\qrL}(W_{\wt{I},q}^{\cC})} & = \quad \text{[Since \(\cl_{\qrL}(W_{\wt{I},q}^{\cC}) = W_{\wt{I},q}^{\cC}\) for every \(q \in \wt{Q}\)]}\\
{\textstyle\bigcup_{q \in \wt{Q}}W_{\wt{I},q}^{\cC}}\enspace .\tag*{\qedhere} 
\end{align*}
\end{proof}

Observe that, since the canonical RFA \(\cC = (\wt{Q}, Σ, \wt{δ}, \wt{I}, \wt{F})\) for a language \(L\) is \emph{strongly consistent}, the left language of each state is a principal of \(\cl_{\qrL}\).
In particular, if the right language of a state is \(u^{-1}L\) then its left language is the principal \(\cl_{\qrL}(u)\).
Therefore, if \(W_{I,q}^{\cN} = \bigcup_{q \in \wt{Q}}W_{\wt{I},q}^{\cC}\) then \(W_{I,q}^{\cN}\) is a closed set in \(\cl_{\qrL}\).
However, the reverse implication does not hold since only the \emph{\(L\)-prime} principals are left languages of states of \(\cC\).

On the other hand, \(L\)-composite principals for \(\qrL\) can be described as intersections of \(L\)-prime principals (see Lemma~\ref{lemma:CompositeIntersection} in Appendix~\ref{sec:supp-results}).
As a consequence, \(\cG{r}(\cN)\) is isomorphic to \(\cC\) if{}f the left languages of states of \(\cN\) are \emph{union of non-empty intersections of left languages of \(\cC\)}, while, as shown by Tamm~\cite{tamm2015generalization}, \(\cN^{\text{res}}\) is isomorphic to \(\cC\) if{}f the left languages of the states of \(\cN\) are \emph{union of left languages of \(\cC\)}.

\section{Learning Residual Automata}\label{sec:LearningNL:qo}
Bollig et al.~\cite{bollig2009angluin} devised the NL\(^*\) algorithm for learning the canonical RFA for a given regular language.
The algorithm describes the behavior of a \emph{Learner} that infers a language \(L\) by performing membership queries on \(L\) (which are answered by a \emph{Teacher}) and equivalence queries between the language accepted by a candidate automaton and \(L\) (which are answered by an \emph{Oracle}).
The algorithm terminates when the \emph{Learner} builds an RFA accepting the language \(L\).
Appendix~\ref{sec:LearningNL} contains a formal description of the NL\(^*\) algorithm.

In this section we present a quasiorder-based perspective on the NL\(^*\) algorithm in which the \emph{Learner} iteratively refines a quasiorder \(\qo\) on \(Σ^*\) by querying the \emph{Teacher} and uses an adaption of the automata construction \(\cH^{r}(\qo,L)\) from Definition~\ref{def:right-const:qo} to build an automaton that is used to query the \emph{Oracle}.
We capture this approach in the so-called \emph{NL\(^{\qo}\) algorithm} whose pseudocode we defer to Appendix~\ref{sec:supp-results}.
Here we give the definitions and general steps of the NL\(^{\qo}\)  algorithm.

The \emph{Learner} maintains a prefix-closed finite set \(\Pref \subseteq Σ^*\) and a suffix-closed finite set \(\Suf \subseteq Σ^*\).
The set \(\Suf\) is used to \emph{approximate} the principals in \(\qrL\) for the words in \(\Pref\).
In order to manipulate these approximations, we define the following two operators.

\begin{definition}\label{def:subsetS}
Let \(L\) be a language, \(\Suf \subseteq \Sigma^*\) and \(u,v \in Σ^*\).
Then \(u^{-1}L =_{\Suf} v^{-1}L  \udiff  \left(u^{-1}L \cap \Suf\right) = \left(v^{-1}L \cap \Suf\right)\).
Similarly, \(u^{-1}L \subseteq_{\Suf} v^{-1}L  \udiff  \left(u^{-1}L \cap \Suf\right) \subseteq \left(v^{-1}L \cap \Suf\right)\).
\end{definition}

These operators allow us to define a version of Nerode's quasiorder restricted to \(\Suf\).

\begin{definition}[Right-language-based quasiorder w.r.t. \(\Suf\)]%
\label{def:finiteNerode}
Let \(L\) be a language, \(\Suf \subseteq \Sigma^*\) and \(u,v \in Σ^*\).
Define \(u \qA v \udiff u^{-1}L \subseteq_{\Suf} v^{-1}L\).
\end{definition}

Recall that the \emph{Learner} only manipulates the principals for the words in \(\Pref\).
Therefore, we need to adapt the notion of composite principal for \(\qA\).

\begin{definition}[\(L_{\Suf}\)-Composite Principal w.r.t. \(\Pref\)]
Let \(\Pref, \Suf \subseteq \Sigma^*\) with \(u \in \Pref\) and let \(L \subseteq Σ^*\) be a language. 
We say that the principal \(\cl_{\qA}(u)\) is \emph{\(L_{\Suf}\)-composite w.r.t. \(\Pref\)} if{}f
\(u^{-1}L =_{\Suf} \bigcup_{x \in \Pref, \; x \qAn u} x^{-1}L\).
Otherwise, we say it is \(L_{\Suf}\)-\emph{prime} w.r.t. \(\Pref\).
\end{definition}

The \emph{Learner} uses the quasiorder \(\qA\) to build an automaton by adapting the construction from Definition~\ref{def:right-const:qo} in order to use only the information that is available by means of the sets \(\Suf\) and \(\Pref\).
Building such an automaton requires the quasiorder to satisfy two conditions: it must be \emph{closed} and \emph{consistent} w.r.t. \(\Pref\).

\begin{definition}[Closedness and Consistency of \(\qA\) w.r.t. \(\Pref\)]\label{def:ClosedCons}\hfill
\begin{alphaenumerate}
\item \(\qA\) is \emph{closed w.r.t. \(\Pref\)} if{}f \(\forall u \in \cP, a \in \Sigma,\;  \cl_{\qA}(ua) \text{ is \(L_{\Suf}\)-prime w.r.t. \(\Pref\)} \Ra \exists v \in \cP,\linebreak  \cl_{\qA}(ua) = \cl_{\qA}(v)\).\label{def:ClosedCons:Closed}
\item \(\qA\) is \emph{consistent w.r.t. \(\Pref\)} if{}f \(\forall u, v \in \Pref, a \in Σ: \; u \qA v \Ra ua \qA va\).\label{def:ClosedCons:Cons}

\end{alphaenumerate}
\end{definition}

At each iteration, the \emph{Learner} checks whether the quasiorder \(\qA\) is closed and consistent w.r.t. \(\Pref\).
If \(\qA\) is not closed w.r.t. \(\Pref\), then it finds \(\cl_{\qA}(ua)\) with \(u \in \Pref, a \in Σ\) such that \(\cl_{\qA}(ua)\) is \(L_{\Suf}\)-prime w.r.t. \(\Pref\) and it is not equal to some \(\cl_{\qA}(v)\) with \(v \in \Pref\).
Then the \emph{Learner} adds \(ua\) to \(\Pref\).

Similarly, if \(\qA\) is not consistent w.r.t. \(\Pref\), the \emph{Learner} finds \(u, v \in \Pref, a \in Σ, x \in \Suf\) such that \(u \qA v\) but \(uax \in L \land vax \notin L\).
Then the \emph{Learner} adds \(ax\) to \(\Suf\).
When the quasiorder \(\qA\) is closed and consistent w.r.t. \(\Pref\), the \emph{Learner} builds the automaton \(\cL(\qA, \Pref)\).

Definition~\ref{def:right-const:qo:S} is an adaptation of the automata construction \(\cH^{r}\) from Definition~\ref{def:right-const:qo}.
Instead of considering all principals, it considers only those that correspond to words in \(\Pref\).
Moreover, the notion of \(L\)-primality is replaced by \(L_{\Suf}\)-primality w.r.t. \(\Pref\), since the algorithm does not manipulate quotients of \(L\) by words in \(Σ^*\) but the approximation through \(\Suf\) of the quotients of \(L\) by words in \(\Pref\) (see Definition~\ref{def:subsetS}).
Note that, if \(\Suf = \Pref = Σ^*\) then \(\cF{r}(L) = \cR(\qA,\Pref)\).

\begin{definition}[Automata construction \(\cL(\qA, \Pref)\)]%
\label{def:right-const:qo:S}
Let \(L\subseteq Σ^*\) be a language and let \(\Pref,\Suf\subseteq Σ^*\).
Define the automaton \(\cL(\qA, \Pref)= (Q, \Sigma, \delta, I, F)\) with \(Q = \{\cl_{\qA}(u) \mid u\in \Pref, \linebreak \cl_{\qA}(u) \text{ is \(L_{\Suf}\)-prime w.r.t. \(\Pref\)}\}\), \(I = \{\cl_{\qA}(u) \in Q \mid \varepsilon \in \cl_{\qA}(u)\}\), \(F = \{\cl_{\qA}(u) \in Q \mid u \in L\}\) and \( \delta(\cl_{\qA}(u), a) = \{ \cl_{\qA}(v) \in Q \mid \cl_{\qA}(u)   a \subseteq \cl_{\qA}(v)\}\) for all \(\cl_{\qA}(u) \in Q\) and \(a \in Σ\).
\end{definition}

Finally, the \emph{Learner} asks the \emph{Oracle} whether \(\lang{\cL(\qA, \Pref)} = L\).
If the \emph{Oracle} answers \emph{yes} then the algorithm terminates.
Otherwise, the \emph{Oracle} returns a counterexample \(w\) for the language equivalence.
Then, the \emph{Learner} adds every suffix of \(w\) to \(\Suf\) and repeats the process.

Theorem~\ref{theorem:NLqo} shows that the NL\(^{\qo}\) algorithm exactly coincides with NL\(^*\).

\begin{restatable}{theoremR}{NLqo}%
\label{theorem:NLqo}
NL\(^{\qo}\) builds the same sets \(\Pref\) and \(\Suf\), performs the same queries to the \emph{Oracle} and the \emph{Teacher} and returns the same RFA as NL\(^*\), provided that both algorithms resolve nondeterminism the same way. 
\end{restatable}

It is worth to remark that, by replacing the right quasiorder \(\qA\) by the right congruence \(\sim_{L_{\Suf}} \ud \mathord{\qA} \cap \mathord{(\qA)^{-1}}\) in the above algorithm (precisely, in Definitions~\ref{def:ClosedCons} and \ref{def:right-const:qo:S}), the resulting algorithm corresponds to Angluin's L\(^*\) algorithm~\cite{angluin1987learning}.
Note that, in that case, all principals \(\cl_{\sim_{L_{\Suf}}}(u)\), with \(u\in\Sigma^*\), are \(L_{\Suf}\)-prime w.r.t. \(\Pref\).

\section{Related Work and Conclusions}

Denis et al.~\cite{denis2002residual} introduced the notion of RFA and canonical RFA for a language and devised a procedure, similar to the subset construction for DFAs, to build the RFA \(\cN^{\text{res}}\) from a given automaton \(\cN\).
Furthermore, they showed that \(\cN^{\text{res}}\) is isomorphic to the canonical RFA \(\cC\) for \(\lang{\cN}\) when \(\cN\) is a co-RFA with no empty states\@.
Later, Tamm~\cite{tamm2015generalization} showed that \(\cN^{\text{res}}\) is isomorphic to \(\cC\) if{}f \emph{the left language of every state of \(\cN\) is a union of left languages of states of \(\cC\)}.
This result generalizes the double-reversal method for building the canonical RFA along the lines of the generalization by Brzozowski and Tamm~\cite{brzozowski2014theory} of the double-reversal method for DFAs, which claims that determinizing an automaton \(\cN\) yields the minimal DFA if{}f \emph{the left language of each state of \(\cN\) is a union of co-atoms of \(\lang{\cN}\)}.
Although the two generalizations have a common foundation, the connection between the two results is not immediate.

Recently~\cite{ganty2019congruence}, we offered a congruence-based perspective of the generalized double-reversal method for DFAs and showed that determinizing an NFA, \(\cN\), yields the minimal DFA for \(\lang{\cN}\) if{}f \(\cl_{\rrL}(W_{I,q}^{\cN}) = W_{I,q}^{\cN}\).
In this paper we extend our previous work and devise quasiorder-based automata constructions that result in RFAs.
One of these constructions, when instantiated with the automata-based quasiorder from Definition~\ref{def:automataQO}, defines a residualization operation that, given an NFA \(\cN\), produces the RFA \(\cG{r}(\cN)\) with, at most, as many states as \(\cN^{\text{res}}\), the residualization operation defined by Denis et al.~\cite{denis2002residual}.
Observe that if \(\cN\) is a co-RFA with no empty states then both \(\cN^{\text{res}}\) and \(\cG{r}(\cN)\) are isomorphic to \(\cC\).

On the other hand, Theorem~\ref{theorem:canonicalreverserestic} shows that \(\cG{r}(\cN)\) is isomorphic to \(\cC\) if{}f\linebreak \(\cl_{\qrL}(W_{I,q}^{\cN}) = W_{I,q}^{\cN}\).
We believe that the similarity between the generalizations of the double-reversal methods for DFAs (\(\cl_{\rrL}(W_{I,q}^{\cN}) = W_{I,q}^{\cN}\)) and for RFAs (\(\cl_{\qrL}(W_{I,q}^{\cN}) = W_{I,q}^{\cN}\)) evidences that quasiorders are for RFAs as congruences are for DFAs.
Indeed, determinizing an NFA \(\cN\) with \(L = \lang{\cN}\) yields the minimal DFA for \(L\) if{}f \(\mathord{\rrN} = \mathord{\rrL}\)~\cite{ganty2019congruence} and, similarly, when residualizing \(\cN\) with our residualization operation we obtain the canonical RFA for \(L\) if{}f \(\mathord{\qrN} = \mathord{\qrL}\), as shown by Lemma~\ref{lemma:qrlEqualqrNResEqualCan}.

It is worth to remark that the left languages of the minimal DFA for \(L\) are principals of \(\rrL\)~\cite{ganty2019congruence}.
Therefore, the condition \(\cl_{\rrL}(W_{I,q}^{\cN}) = W_{I,q}^{\cN}\), which guarantees that determinizing \(\cN\) yields the minimal DFA, can be stated as: \emph{the left language of each state of \(\cN\) is a union of left languages of states of the minimal DFA}.
Thus, this characterization is the DFA-equivalent of Tamm's condition~\cite{tamm2015generalization} for RFAs.%

Figure~\ref{fig:conclusions} summarizes the existing results about these double-reversal methods.

\begin{figure}[t]
\begin{minipage}[l]{0.63\textwidth}
\begin{tabu}{@{}c|c@{}}
   \begin{tabular}{c}\textbf{Brzozowski}
   \textbf{and Tamm}\end{tabular}\cite{brzozowski2014theory}& \textbf{Ganty et al.~\cite{ganty2019congruence}}\\[3pt]  
   \begin{tabular}{c} \(\cN^{D} \equiv \cM\) \\ if{}f \\ \(\forall q, W_{I,q}^{\cN} \text{ is a union of co-atoms} \)\end{tabular} & \begin{tabular}{c}\(\cN^{D} \equiv \cM\) \\ if{}f \\ \(\forall q, \cl_{\rrL}(W_{I,q}^{\cN}) = W_{I,q}^{\cN}\)\end{tabular}  \\
   \vspace{-7pt}\\
   \tabucline[1pt off 2pt]
   \\
   \vspace{-7pt}\\
   \textbf{Tamm~\cite{tamm2015generalization}} & \textbf{Theorem~\ref{theorem:canonicalreverserestic}}\\[3pt] 
   \begin{tabular}{c}\(\cN^{\text{res}} \equiv \cC\) \\ if{}f \\ \(\forall q, W_{I,q}^{\cN} \text{ is a union of }  W_{I,q'}^{\cC}\)\end{tabular} & \begin{tabular}{c} \(\cG{r}(\cN) \equiv \cC\) \\ if{}f \\ \( \forall q, \cl_{\qrL}(W_{I,q}^{\cN}) = W_{I,q}^{\cN}\)\end{tabular} \\
\end{tabu}
\end{minipage}\hfill%
\begin{minipage}[r]{0.34\textwidth}

In the diagram: \(\cN\) is an NFA with \(L = \lang{\cN}\); \(\cN^{D}\) is the result of determinizing \(\cN\) with the standard subset construction; \(\cM\) is the minimal DFA for \(L\); \(\cC = \cF{r}(L)\) is the canonical RFA for \(L\) and \(\cN_1 \equiv \cN_2\) denotes that automaton \(\cN_1\) is isomorphic to \(\cN_2\).
\end{minipage}%
\caption{Summary of the existing results about the generalized double-reversal method for building the minimal DFA (first row) and the canonical RFA (second row) for a given language. The results on the first column are based on the notion of \emph{atoms} of a language while the results on the second column are based on \emph{quasiorders}.}%
\label{fig:conclusions}  
\end{figure}

Moreover, we support the idea that quasiorders are natural to residual automata by observing that the NL\(^*\) algorithm can be interpreted as an algorithm that, at each iteration, refines an approximation of the Nerode's quasiorder and builds an RFA using our automata construction.

Finally, it is worth to mention that Myers et al.~\cite{MyersAMU15} describe different canonical nondeterminism automata constructions for a given regular language and show how to obtain the canonical RFA.
They do it by first constructing the minimal DFA for the language interpreted in a variety of join-semilattices and then applying a dual equivalence between this variety and the category of closure spaces.
In some sense, this already establishes a connection between the class of DFAs and RFAs.
Indeed, the same authors~\cite{AdamekMUM14} use this category-theoretical perspective to address the residual-equivalent of the double-reversal method proposed by Denis et al.~\cite{denis2002residual}.
In contrast, this work revisit different methods to construct the canonical RFA relying on the simple notion of quasiorders on words, as a natural extension of our work on congruences for the study of minimization techniques for DFAs.

\appendix

\section{Learning Algorithm \texorpdfstring{\textit{NL}\(^*\)}{NL*}}%
\label{sec:LearningNL}
Bollig et al.~\cite{bollig2009angluin} devised an algorithm, NL\(^*\), that learns the canonical RFA for a given regular language \(L\).
Similarly to the well-known L\(^*\) algorithm of Angluin~\cite{angluin1987learning}, the NL\(^*\) algorithm relies on a \emph{Teacher}, which answers membership queries for \(L\), and an \emph{Oracle} which answers equivalence queries between the language accepted by an RFA and \(L\).

The \emph{Learner} maintains a prefix-closed finite set \(\Pref \subseteq Σ^*\) and a suffix-closed finite set \(\Suf \subseteq Σ^*\).
The \emph{Learner} groups the words in \(\Pref\) by building a table \(T = (\cT, \Pref, \Suf)\) where \(T: (\Pref \cup \Pref   Σ) \times \Suf \to \{{+}, {-}\}\) is a function such that for every \(u \in \Pref \cup \Pref   Σ\) and \(v \in \Suf\) we have that \(T(u,v) = {+} \Lra uv \in L\).
Otherwise \(T(u,v) = {-}\).
For every word \(u \in \Pref \cup \Pref   Σ\), define the function \(\row(u): \Suf \to \{{+}, {-}\}\) as \(\row(u)(v) \ud T(u,v)\).
The set of all rows of a table \(\cT\) is denoted by \(\Rows(\cT)\).

The algorithm uses the table \(\cT = (T, \Pref, \Suf)\) to build an automaton whose states are some of the rows \(\cT\). 
In order to do that, it is necessary to define the notions of \emph{union} of rows, \emph{prime} row and \emph{composite} row.

\begin{definition}[Join Operator]%
\label{def:join}
Let \(\cT = (T, \Pref, \Suf)\) be a table.
For every \(r_1, r_2 \in \Rows(\cT)\), define the \emph{join} \(r_1 \sqcup r_2: \Suf \to \{{+},{-}\}\) as:
\[\forall x \in \Suf, \; (r_1 \sqcup r_2)(x) \ud \left\{\begin{array}{ll}
{+} & \text{if } r_1(x) = {+} \lor r_2(x) = {+} \\
{-} & \text{otherwise}\end{array}\right.\]
\end{definition}

Note that the join operator is associative, commutative and idempotent.
However, the join of two rows is not necessarily a row of \(\cT\).

\begin{definition}[Covering Relation]%
\label{def:coverRow}
Let \(\cT = (T, \Pref, \Suf)\) be a table.
Then, for every pair of rows \(r_1, r_2 \in \Rows(\cT)\) we have that
\(r_1 \sqsubseteq r_2 \udiff \forall x \in \Suf, \; r_1(x) = {+} \Ra r_2(x) = {+}\).
We write \(r_1 \sqsubset r_2\) to denote \(r_1 \sqsubseteq r_2 \) and \(r_1 \neq r_2\).
\end{definition}

\begin{definition}[Composite and Prime Rows]%
\label{def:PrimeRow}
Let \(\cT = (T, \Pref, \Suf)\) be a table.
We say a row \(r \in \Rows(\cT)\) is \(\cT\)-\emph{composite} if it is the join of all the rows that it strictly covers, i.e., \(r = \bigsqcup_{r' \in \Rows(\cT), \; r' \sqsubset r} r'\).
Otherwise, we say \(r\) is \(\cT\)-\emph{prime}.
\end{definition}

\begin{definition}[Closed and Consistent Table]%
\label{def:Table}
Let \(\cT = (T, \Pref, \Suf)\) be a table.
Then
\begin{alphaenumerate}
\item \(\cT\) is \emph{closed} if \(\forall u \in \Pref, a \in Σ, \; \row(ua) = \bigsqcup \{\row(v) \mid v \in \Pref, \; \row(v) \sqsubseteq \row(ua) \land \row(v) \text{ is \(\cT\)-prime}\}\).%
\label{def:Table:closed}
\item \(\cT\) is \emph{consistent} if \(\row(u) \sqsubseteq \row(v) \Ra \row(ua) \sqsubseteq \row(va) \) for every \(u,v \in \Pref\) and \(a \in Σ\).%
\label{def:Table:Consistent}
\end{alphaenumerate}
\end{definition}

At each iteration of the algorithm, the \emph{Learner} checks whether the table \(\cT = (T, \Pref, \Suf)\) is closed and consistent.
If \(\cT\) is not closed, then it finds \(\row(ua)\) with \(u \in \Pref, a \in Σ\) such that \(\row(ua)\) is \(\cT\)-prime and it is not equal to some \(\row(v)\) with \(v \in \Pref\).
Then the \emph{Learner} adds \(ua\) to \(\Pref\) and updates the table \(\cT\).
Similarly, if \(\cT\) is not consistent, the \emph{Learner} finds \(u, v \in \Pref, a \in Σ, x \in \Suf\) such that \(\row(u) \subseteq \row(v)\) but \(\row(ua)(x) = {+} \land \row(va)(x) = {-}\).
Then the \emph{Learner} adds \(ax\) to \(\Suf\) and updates \(\cT\).
When the table \(\cT\) is closed and consistent, the \emph{Learner} builds the RFA \(\cR(\cT)\).

\begin{definition}[\(\cR(\cT)\)]
Let \(\cT = (T, \Pref, \Suf)\) be a table.
Define the automaton \(\cR(\cT) = (Q, Σ, I, F, δ)\) with \(Q = \{\row(u) \mid u \in \Pref \land \row(u) \text{ is \(\cT\)-prime}\}\), \(I = \{\row(u) \in Q \mid \row(u) \sqsubseteq \row(\varepsilon)\}\), \linebreak\(F = \{\row(u) \in Q\mid \row(u)(\varepsilon) = {+}\}\) and \(\row(v) \in δ(\row(u),a) = \{\row(v) \in Q \mid \row(v) \sqsubseteq \row(ua)\}\) for all \(\row(u) \in Q, a \in Σ\). 
\end{definition}

The \emph{Learner} asks the \emph{Oracle} whether \(\lang{\cR(\cT)} = L\).
If the \emph{Oracle} answers \emph{yes} then the algorithm terminates.
Otherwise, the \emph{Oracle} returns a counterexample \(w\) for the language equivalence.
Then the \emph{Learner} adds every suffix of \(w\) to \(\Suf\), updates the table \(\cT\) and repeats the process.

\section{Supplementary Results}
\label{sec:supp-results}
In this section we include auxiliary results that we refer to in the main part of the document and/or we use in the deferred proofs.

The following result establishes a relationship between the \(L\)-composite principals for two comparable right quasiorders \(\mathord{\qr_1} \subseteq \mathord{\qr_2}\).
This result is used in Theorem~\ref{theorem:numLPrimePrincipals} to show that the number of \(L\)-prime principals induced by \(\qr_1\) is greater than or equal to the number of \(L\)-prime principals induced by \(\qr_2\).

\begin{lemma}\label{lemma:numprincipals}
Let \(L \subseteq Σ^*\) be a regular language and let \(u \in Σ^*\). 
Let \(\qr_1\) and \(\qr_2\) be two right \(L\)-preserving quasiorders  such that \(\mathord{\qr_1} \subseteq \mathord{\qr_2}\).
Then
\[\cl_{\qr_1}(u) \text{ is \(L\)-composite} \Ra \left( \cl_{\qr_2}(u) \text{ is \(L\)-composite} \lor  \exists x \qrn_1 u, \; \cl_{\qr_2}(u) = \cl_{\qr_2}(x)\right)\enspace .\]
Similarly holds for left \(L\)-preserving quasiorders.
\end{lemma}
\begin{proof}
Let \(u \in Σ^*\) be such that \(\cl_{\qr_1}(u)\) is \(L\)-composite.
Then \(u^{-1}L = \bigcup_{x \in Σ^*, x \qrn_1 u} x^{-1}L\).
On the other hand, since \(\qr_2\) is a right \(L\)-preserving quasiorder, we have that \(\mathord{\qr_2} \subseteq \mathord{\qrL}\), as shown de Luca and Varricchio~\cite{de1994well}.
Therefore \(u^{-1}L \supseteq \bigcup_{x \in Σ^*, x \qrn_2 u} x^{-1}L\). 
There are now two possibilities:
\begin{itemize}
\item For all \(x \in Σ^*\) such that \(x \qrn_1 u\) we have that \(x \qrn_2 u\).
In that case we have that \(u^{-1}L = \bigcup_{x\inΣ^*, \; x \qrn_2 u} x^{-1}L\), hence \(\cl_{\qr_2}(u)\) is \(L\)-composite.
\item There exists \(x \in Σ^*\) such that \(x \qrn_1 u\), hence \(x \qr_2 u\), but \(x \not\qrn_2u\).
In that case, it follows that \(\cl_{\qr_2}(x) = \cl_{\qr_2}(u)\).\qedhere
\end{itemize}
\end{proof}

The following lemma allows us to conclude that \(\cF{r}(L)\) is invariant to our residualization operation \(\cG{r}\).
\begin{lemma}
\label{lemma:qrHEqualqrifHsc}
Let \(L\) be a regular language and let \(\qr\) be a right quasiorder such that \(\cl_{\qr}(L) = L\).
Let \(\cH = \cH^{r}(\qr,L)\).
If \(\cH\) is a strongly consistent RFA then \(\mathord{\qr_{\cH}} = \mathord{\qr}\).
\end{lemma}

\begin{proof}
Let \(\cN = (Q, Σ, δ, I, F)\) and \(\cH = (\wt{Q}, Σ, \wt{δ}, \wt{I}, \wt{F})\).
As shown by Lemma~\ref{lemma: HrGeneratesL}, \(\cH = \cH^{r}(\qr, L)\) is an RFA accepting \(L\), hence each state of \(\cH\) is an \(L\)-prime principal \(\cl_{\qr}(u)\) whose right language is the quotient \(u^{-1}L\) for some \(u \in Σ^*\).

Observe that \(\mathord{\qr_{\cH}} = \mathord{\qr}\) holds if{}f for every \(u,v \in Σ^*, \; \post_u^{\cH}(\wt{I}) \subseteq \post_v^{\cH}(\wt{I}) \Lra u \qr v\).
Next we show that:
\begin{equation}%
\label{eq:PostQuotient}
\post_u^{\cH}(\wt{I}) = \{\cl_{\qr}(x) \in \wt{Q} \mid x \qr u\} \enspace .
\end{equation}
First, we prove that \(\post_u^{\cH}(\wt{I}) \subseteq \{\cl_{\qr}(x) \in \wt{Q} \mid x \qr u\}\). Let \(\cl\) denote \(\cl_{\qr}\).
\begin{align*}
\cl(x) \in \post_u^{\cH}(\wt{I}) & \Lra \quad \text{[By definition of \(\post_u^{\cH}(\wt{I})\)]}\\
\exists \cl(x_0) \in \wt{I}, \; u \in W^{\cH}_{\cl(x_0),\cl(x)} & \Ra \quad \text{[By Definition~\ref{def:right-const:qo}]} \\
\exists \cl(x_0) \in \wt{Q}, \;\varepsilon \in \cl(x_0) \land \cl(x_0)\, u \subseteq \cl(x) & \Lra \quad \text{[By definition of \(\cl\)]} \\
\exists \cl(x_0) \in \wt{Q}, \; x_0 \qr \varepsilon \land x \qr x_0u & \Ra \quad \text{[Since \(x_0 \qr \varepsilon \Ra x_0 u \qr u\)]} \\
x \qr u \enspace .
\end{align*}

We now prove the reverse inclusion. 
Let \(\cl(u), \cl(x) \in \wt{Q}\) be such that \(x \qr u\).
Then,
\begin{align*}
\cl(u) \in \wt{Q} & \Ra \; \text{[By Lemma~\ref{lemma: HrGeneratesL}]}\\
W^{\cH}_{\cl(u),F} = u^{-1}L & \Ra \; \text{[Since \(\cH\) is strly. consistent]} \\
u \in W_{I,\cl(u)}^{\cH} & \Ra \; \text{[By def. \(W_{S,T}^{\cH}\) with \(u = za\)]} \\
\exists \cl(y) \in \wt{Q}, \cl(u_0) \in \wt{I},\; z \in W_{\cl(u_0), q'}  \land a \in W_{\cl(y), \cl(u)}& \Ra \; \text{[By Definition~\ref{def:right-const:qo}]} \\
\exists \cl(y) \in \wt{Q}, \cl(u_0) \in \wt{I},\;\cl(u_0) z \subseteq \cl(y) \land \cl(y) a \subseteq \cl(u)& \Ra \; \text{[By definition of \(\cl = \cl_{\qr}\)]} \\
\exists \cl(y) \in \wt{Q}, \cl(u_0) \in \wt{I},\;\cl(u_0) z \subseteq \cl(y) \land u\qr y a & \Ra \; \text{[Since \(x \qr u\)]} \\
\exists \cl(y) \in \wt{Q}, \cl(u_0) \in \wt{I},\;\cl(u_0) z \subseteq \cl(y)\land x \qr y a & \Ra \; \text{[By definition of \(\cl = \cl_{\qr}\)]} \\
\exists \cl(y) \in \wt{Q}, \cl(u_0) \in \wt{I},\;\cl(u_0) z \subseteq \cl(y)\land \cl(y)a \subseteq \cl(x) & \Ra \; \text{[By definition of \(\post_u^{\cH}(\wt{I})\)]} \\
\cl(x) \in \post_u(I) \enspace .
\end{align*}

It follows from Equation~\eqref{eq:PostQuotient} that \(\post_u^{\cH}(I) \subseteq \post^{\cH}_v(I) \Lra u \qr v\), i.e., \(\mathord{\qr_{\cH}} = \mathord{\qr}\).\qedhere
\end{proof}
The following lemma shows that, if we consider the right Nerode's quasiorder \(\qrL\) then composite principals can be described as intersections of prime principals.

\begin{lemma}\label{lemma:CompositeIntersection}
Let \(\cN = (Q, Σ, δ, I, F)\) be an NFA with \(\lang{\cN} = L\).
Then,
\begin{equation}\label{eq:rhoCompRaIntersection}
u^{-1}L = \hspace{-10pt}\bigcup_{x\in\Sigma^*,\; x \qrn_L u}\hspace{-10pt} x^{-1}L \implies \cl_{\qrL}(u) = \hspace{-10pt}\bigcap_{x\in\Sigma^*,\; x \qrn_L u}\hspace{-10pt} \cl_{\qrL}(x) \enspace .
\end{equation}
\end{lemma}
\begin{proof}
Observe that the inclusion \(\cl_{\qrL}(u) \subseteq \bigcap_{x\in\Sigma^*, x \qrn_L u} \cl_{\qrL}(x)\) always holds since\linebreak \(x \qrn_L u \Ra \cl_{\qrL}(u) \subseteq \cl_{\qrL}(x)\).
Next, we show the reverse inclusion.

Assume that the left hand side of Equation~\eqref{eq:rhoCompRaIntersection} holds and let \(w \in \bigcap_{x\in\Sigma^*, x \qrn_L u} \cl_{\qrL}(x)\).
Then, by definition of intersection and \(\cl_{\qrL}\), we have that \(x \qrL w\) for every \(x \in Σ^*\) such that \(x \qrn_L u\), i.e., \(x^{-1}L \subseteq w^{-1}L\) for every word \(x \in Σ^*\) such that \(x^{-1}L \subsetneq u^{-1}L\).
Since, by hypothesis, \(u^{-1}L = \bigcup_{x\in\Sigma^*,\; x \qrn_L u} x^{-1}L\), it follows that \(u^{-1}L \subseteq w^{-1}L\) and, therefore, \(w \in \cl_{\qr}(u)\).
We conclude that \(\bigcap_{x\in\Sigma^*, x \qrn_L u} \cl_{\qrL}(x) \subseteq \cl_{\qrL}(u)\).
\end{proof}

Finally, we show algorithm NL\(^\qo\), the quasiorder-based version of the algorithm NL\(^*\).

\RemoveAlgoNumber
\begin{algorithm}[t]
\SetAlgorithmName{Algorithm NL\(^{\qo}\)}{}
\SetSideCommentRight%
\caption{Quasiorder-based version of NL\(^*\)}\label{alg:NL-star}
\KwData{A \emph{Teacher} that answers membership queries in \(L\)}
\KwData{An \emph{Oracle} that answers equivalence queries between the language accepted by an RFA and \(L\)}
\KwResult{The canonical RFA for the language \(L\).}
\(\cP, \cS := \{\varepsilon\}\)\;
\While{True\label{step:teacher-yes}}{%
\label{step:loop}
	\While{\(\qA\) not closed or consistent:}{
		\If{\(\qA\) is not closed}{
			Find \(u \in \Pref, a \in \Sigma\) with \(\cl_{\qA}(u)\) \(L_{\Suf}\)-prime for \(\Pref\) and \(\forall v \in \Pref, \; \cl_{\qA}(u) \neq \cl_{\qA}(v)\)\;
			Let \(\cP := \cP \cup \{ua\}\)\;
		}
		\If{\(\qA\) is not consistent}{
			Find \(u, v \in \Pref, a \in \Sigma\) with \(u \qA v\) s.t. \(u a \not\qA v a  \)\;
			Find \(x \in ((ua)^{-1}{L}\cap \Suf) \cap ((va)^{-1}{L} \cap \Suf)^c\) \;
			Let \(\cS := \cS \cup \{ax\}\)\;
		}
	}
	\label{step:DFA-const}Build \(\cL(\qA, \Pref)\)\;
	Ask the \emph{Oracle} whether \(L = \lang{\cL(\qA,\Pref)}\)\;
	\If{the \emph{Oracle} replies with a counterexample \(w\)}{
	Let \(\Suf := \Suf \cup \{x \in \Sigma^* \mid w = w'x \text{ with }w \in \Suf, w' \in \Sigma^*\}\)\;
	}\Else{\Return{\(\cL(\qA, \Pref)\)}\;}
} 
\end{algorithm}

\section{Deferred Proofs}%
\label{sec:proofs}

\CongruencePbwComplete*
\begin{proof}\hfill
\begin{enumerate}
\item \(\qr\) is a right quasiorder if{}f \(\cl_{\qr}(v)u \subseteq \cl_{\qr}(vu)\), for all \(u,v \in \Sigma^*\).

(\(\Ra\)).
Let \(x \in \cl_{\qr}(v)u\).
Then, \(x = \tilde{v}u\) with \(v \qr \tilde{v}\). 
Since \(\qr\) is a right quasiorder and \(v \qr \tilde{v}\) then \(vu \qr \tilde{v}u\).
Therefore \(x \in \cl_{\qr}(vu)\).
\\(\(\La\)).
Assume that for each \(u,v \in \Sigma^*\) and \(\tilde{v} \in \cl_{\qr}(v)\) we have that \(\tilde{v}u \in \cl_{\qr}(vu)\).
Then, \(v \qr \tilde{v} \Ra vu \qr \tilde{v}u\).

\item\(\ql\) is a left quasiorder if{}f \(u\cl_{\ql}(v) \subseteq \cl_{\ql}(uv)\), for all \(u,v \in \Sigma^*\).

(\(\Ra\)).
Let \(x \in u\cl_{\ql}(v)\).
Then, \(x = u\tilde{v}\) with \(v \ql \tilde{v}\). 
Since \(\ql\) is a left quasiorder and \(v \ql \tilde{v}\) then \(uv \ql u\tilde{v}\).
Therefore \(x \in \cl_{\ql}(uv)\).
\\(\(\La\)).
Assume that for each \(u,v \in \Sigma^*\) and \(\tilde{v} \in \cl_{\ql}(v)\) we have that \(u\tilde{v} \in \cl_{\ql}(uv)\).
Then \(v \ql \tilde{v} \Ra uv \ql u\tilde{v}\).\qedhere
\end{enumerate}
\end{proof}

\HrGeneratesL*
\begin{proof}
To simplify the notation, we denote \(\cl_{\qr}\), the closure induced by the quasiorder \(\qr\), simply by \(\cl\).
Let \(\mathcal{H} = \cH^{r}(\qr, L) = (Q, \Sigma, \delta, I, F)\).
We first show that \(\mathcal{H}\) is an RFA.
\begin{equation}\label{eq:right-langsUco}
W_{\cl(u), F}^{\mathcal{H}} = u^{-1}L, \quad \text{for each } \cl(u)\in Q \enspace .
\end{equation}

Let us prove that \(w \in u^{-1}L \Ra w \in W_{\cl(u), F}^{\mathcal{H}}\).
We proceed by induction on \(|w|\).
\begin{itemize}
\item \emph{Base case:}
Assume \(w = \varepsilon\).
Then, \(\varepsilon \in u^{-1}L \Ra u \in L \Ra \cl(u) \in F \Ra \varepsilon \in W_{\cl(u), F}^{\mathcal{H}}\).

\item \emph{Inductive step:}
Assume that the hypothesis holds for each \(x \in \Sigma^*\) with \(\len{x} \leq n ~(n \geq 1)\), and let \(w \in \Sigma^*\) be such that \(\len{w} = n{+}1\).
Then \(w = a  x\) with \(\len{x} = n\) and \(a \in Σ\).
\begin{align*}
ax \in u^{-1} L & \Ra \quad \text{[By definition of quotient]} \\
x \in (ua)^{-1}L & \Ra \quad \\
\hspace{-15pt}\span \text{[By Defs.~\ref{def:CompositeClosed} and~\ref{def:right-const:qo}, \(\cl(ua)\) is \(L\)-prime (so \(z \ud ua\)) or \((ua)^{-1}L = \hspace{-5pt}\bigcup_{x_i \qrn ua}\hspace{-5pt}x_i^{-1}L\) (so \(z \ud x_i\))]}\\
\exists \cl(z) \in Q, \; x \in z^{-1}L \land \cl(ua) \subseteq \cl(z) & \Ra \quad \text{[By I.H., \(\cl(u) a \subseteq \cl(ua)\) and Def.~\ref{def:right-const:qo}]} \\
x \in W_{\cl(z), F}^{\mathcal{H}} \land \cl(z) \in δ(\cl(u), a) & \Ra \quad \text{[By definition of \(W_{S,T}\)]}\\
ax \in W_{\cl(u), F}^{\mathcal{H}} \enspace .
\end{align*}
\end{itemize}

We now prove the other side of the implication, i.e., \(w \in W_{\cl(u), F}^{\mathcal{H}} \Ra w \in u^{-1}L\).
\begin{itemize}
\item \emph{Base case:}
Let \(w = \varepsilon\).
By Definition~\ref{def:right-const:qo}, \(\varepsilon \in W_{\cl(u), F}^{\mathcal{H}} \Ra \exists \cl(x) \in Q, \; x \in L \land \cl(u)  \varepsilon \subseteq \cl(x)\).
Since \(\cl(L) = L\), we have that \(u\,\varepsilon \in L \), hence \(\varepsilon \in u^{-1}L\).

\item \emph{Inductive step:}
Assume the hypothesis holds for each \(x \in \Sigma^*\) with \(\len{x} \leq n  ~(n \geq 1)\) and let \(w \in \Sigma^*\) be such that \(\len{w} = n{+}1\).
Then \(w = a  x\) with \(\len{x} = n\) and \(a \in Σ\).
\begin{align*}
ax \in W_{\cl(u), F}^{\mathcal{H}} & \Ra \quad \text{[By Definition~\ref{def:right-const:qo}]} \\
x \in W_{\cl(y), F}^{\mathcal{H}} \land \cl(u) a \subseteq \cl(y)  & \Ra \quad \text{[By I.H. and since \(\cl\) is induced by \(\qr\)]} \\
x \in y^{-1}L \land y \qr ua & \Ra \quad \text{[Since \(u \qr v \Ra u^{-1}L \subseteq v^{-1}L\)~\cite{de1994well}\,]} \\
x \in y^{-1}L \land y^{-1}L \subseteq (ua)^{-1}L & \Ra \quad \text{[Since \(x \in (ua)^{-1} L \Ra ax \in u^{-1}L\)]} \\
ax \in u^{-1}L \enspace .
\end{align*}
\end{itemize}

We have shown that \(\mathcal{H}\) is an RFA. 
Finally, we show that \(\lang{\mathcal{H}} = L\). 
First note that,
\[\lang{\mathcal{H}} = \bigcup_{\cl(u) \in I} W_{\cl(u), F}^{\mathcal{H}} = \bigcup_{\cl(u) \in I} u^{-1}L \enspace ,\]
where the first equality holds by definition of \(\lang{\mathcal{H}}\) and the second by Equation~\eqref{eq:right-langsUco}.
On one hand, we have that \(\bigcup_{\cl(u) \in I} u^{-1}L \subseteq L\) since, by Definition~\ref{def:right-const:qo}, \(\varepsilon \in \cl(u)\), for each \(\cl(u) \in I\), and therefore  \(u \qr \varepsilon\) which, as shown by de Luca and Varricchio~\cite{de1994well}, implies that \(u^{-1}L \subseteq \varepsilon^{-1}L = L\).
Let us show that \(L \subseteq \bigcup_{\cl(u) \in I} u^{-1}L\).
First, let us assume that \(\cl(\varepsilon) \in I\).
Then, 
\[L = \varepsilon^{-1}L \subseteq \bigcup_{\cl(u) \in I} u^{-1}L \enspace .\]
Now suppose that \(\cl(\varepsilon)\notin I\), i.e., \(\cl(\varepsilon)\) is \(L\)-composite.
Then 
\[ L = \varepsilon^{-1}L = \bigcup_{u \qrn \varepsilon} u^{-1}L = \bigcup_{\cl(u) \in I} u^{-1}L \enspace .\]
where the last equality follows from \(\cl(u) \in I \Lra \varepsilon \in \cl(u)\).
\end{proof}

\HlgeneratesL*
\begin{proof}
To simplify the notation we denote \(\cl_{\ql}\), the closure induced by the quasiorder \(\ql\), simply by \(\cl\).
Let \(\mathcal{H} = \cH^{\ell}(\ql, L) = (Q, \Sigma, \delta, I, F)\).
We first show that \(\mathcal{H}\) is a co-RFA.
\begin{equation}\label{eq:left-langsUco}
W_{I, \cl(u)}^{\mathcal{H}} = Lu^{-1}, \quad \text{for each } \cl(u)\in Q \enspace .
\end{equation}

Let us prove that \(w \in Lu^{-1} \Ra w \in W_{I, \cl(u)}^{\mathcal{H}}\).
We proceed by induction.
\begin{itemize}
\item \emph{Base case:}
Let \(w = \varepsilon\).
Then, \(\varepsilon \in Lu^{-1} \Ra u \in L \Ra \cl(u) \in I \Ra \varepsilon \in W_{I, \cl(u)}^{\mathcal{H}}\).

\item \emph{Inductive step:}
Assume the hypothesis holds for all \(x \in \Sigma^*\) with \(\len{x} \leq n  ~(n \geq 1)\) and let \(w \in \Sigma^*\) be such that \(\len{w} = n{+}1\).
Then \(w = x a\) with \(\len{x} = n\) and \(a \in Σ\).
\begin{align*}
xa \in Lu^{-1} & \Ra \quad \text{[By definition of quotient]} \\
x \in L(au)^{-1} & \Ra \quad  \\
\hspace{-15pt}\span \text{[By Defs.~\ref{def:CompositeClosed} and~\ref{def:left-const:qo}, \(\cl(au)\) is \(L\)-prime (so \(z \ud au\)) or \(L(au)^{-1} = \hspace{-5pt}\bigcup_{x_i \qln au}\hspace{-5pt}Lx_i^{-1}\) (so \(z \ud x_i\))]}\\
\exists \cl(z) \in Q, \; x \in L z^{-1} \land \cl(au) \subseteq \cl(z) & \Ra \quad \text{[By I.H., \(a \cl(u) \subseteq \cl(au)\) and Def.~\ref{def:left-const:qo}]} \\
x \in W_{I, \cl(z)}^{\mathcal{H}} \land \cl(u) \in δ(\cl(z), a) & \Ra \quad \text{[By definition of \(W_{S,T}\)]}\\
xa \in W_{I, \cl(u)}^{\mathcal{H}} \enspace .
\end{align*}
\end{itemize}

We now prove the other side of the implication, i.e., \(w \in W_{I, \cl(u)}^{\mathcal{H}} \Ra w \in Lu^{-1}\).

\begin{itemize}
\item \emph{Base case:}
Let \(w = \varepsilon\).
Then \(\varepsilon \in W_{I, \cl(u)}^{\mathcal{H}} \Ra \exists \cl(x) \in Q,\; x \in L \land \varepsilon  \cl(u) \subseteq \cl(x)\).
Since \(\cl(L) = L\), we have that \( \varepsilon  u \in L\), hence \(\varepsilon \in Lu^{-1}\).

\item \emph{Inductive step:}
Assume the hypothesis holds for all \(x \in \Sigma^*\) with \(\len{x} \leq n\) and let \(w \in \Sigma^*\) be such that \(\len{w} = n{+}1\).
Then \(w = x a\) with \(\len{x} = n\) and \(a \in Σ\).
\begin{align*}
xa \in W_{I, \cl(u)}^{\mathcal{H}} & \Ra \quad \text{[By Definition~\ref{def:left-const:qo}]} \\
a \cl(u) \subseteq \cl(y) \land x \in W_{I, \cl(y)}^{\mathcal{H}} & \Ra \quad \text{[By I.H. and since \(\cl\) is induced by \(\ql\)]} \\
y \ql au \land x \in Ly^{-1} & \Ra \quad \text{[Since \(u \ql v \Ra Lu^{-1} \subseteq Lv^{-1}\)~\cite{de1994well}\,]} \\ 
Ly^{-1} \subseteq L(au)^{-1} \land x \in Ly^{-1} & \Ra \quad \text{[Since \(x \in L(au)^{-1} \Ra xa \in Lu^{-1}\)]} \\
xa \in u^{-1}L \enspace .
\end{align*}
\end{itemize}

We have shown that \(\mathcal{H}\) is a co-RFA.
Finally, we show that \(\lang{\mathcal{H}} = L\). 
First note that,
\[\lang{\mathcal{H}} = \bigcup_{\cl(u) \in F} W_{I, \cl(u)}^{\mathcal{H}} = \bigcup_{\cl(u) \in F} Lu^{-1} \enspace ,\]
where the first equality holds by definition of \(\lang{\mathcal{H}}\) and the second by Equation~\eqref{eq:left-langsUco}.
On one hand, we have that \(\bigcup_{\cl(u) \in F} Lu^{-1} \subseteq L\) since, by Definition~\ref{def:left-const:qo}, \(\varepsilon \in \cl(u)\), for each \(\cl(u) \in F\), and therefore  \(u \ql \varepsilon\) which, as shown by de Luca and Varricchio~\cite{de1994well}, implies that \(Lu^{-1} \subseteq L\varepsilon^{-1} = L\).
Let us show that \(L \subseteq \bigcup_{\cl(u) \in F} Lu^{-1}\).
First, suppose that \(\cl(\varepsilon) \in F\).
Then, 
\[L = L\varepsilon^{-1} \subseteq \bigcup_{\cl(u) \in F} Lu^{-1} \enspace .\]
Now suppose that \(\cl(\varepsilon)\notin F\), i.e., \(\cl(\varepsilon)\) is \(L\)-composite.
Then
\[ L = L\varepsilon^{-1} = \bigcup_{u \qln \varepsilon} Lu^{-1} = \bigcup_{\cl(u) \in F} u^{-1}L\enspace .\]
where the last equality follows from \(\cl(u) \in F \Lra \varepsilon \in \cl(u)\).
\end{proof}

\leftRightReverse*
\begin{proof}
Let \(\cH^{r}(\qr, L) = (Q, \Sigma, \delta, I, F)\) and \((\cH^{\ell}(\ql, L^R))^R = (\wt{Q}, \Sigma, \wt{\delta}, \wt{I}, \wt{F})\).
We will show that \(\cH^{r}(\qr, L)\) is isomorphic to \((\cH^{\ell}(\ql, L^R))^R\).

Let \(\varphi: Q \rightarrow \wt{Q}\) be a mapping assigning to each state \(\cl_{\qr}(u) \in Q\) with \(u \in \Sigma^*\), the state \(\cl_{\ql}(u^R) \in \wt{Q}\).
We show that \(\varphi\) is an NFA isomorphism between \(\cH^{r}(\qr, L)\) and \((\cH^{\ell}(\ql, L^R))^R\).
Observe that:
\begin{align*}
u^{-1}L = \bigcup_{x \qrn u}x^{-1}L & \Lra \quad \text{[Since \(\big(\bigcup S_i\big)^R = \bigcup S_i^R\)]} \\
(u^{-1}L)^R = \bigcup_{x \qrn u}(x^{-1}L)^R & \Lra \quad \text{[Since \((u^{-1}L)^R = L^R(u^R)^{-1} \)]}\\
L^R(u^R)^{-1} = \bigcup_{x \qrn u} L^R(x^R)^{-1} & \Lra \quad \text{[By hypothesis, \(u \qrn v \Lra u^R \qrn v^R\)]}\\
L^R(u^R)^{-1} = \bigcup_{x^R \qln u^R} L^R(x^R)^{-1} \enspace . 
\end{align*}

It follows that \(\cl_{\qr}(u)\) is \(L\)-composite if{}f \(\cl_{\ql}(u^R)\) is \(L^R\)-composite, hence \(\varphi(Q) = \wt{Q}\).

Since \(\varepsilon \in \cl_{\qr}(u) \Lra u \qr \varepsilon \Lra u^r \ql \varepsilon \Lra \varepsilon \in \cl_{\ql}(u^R)\), we have that \(\cl_{\qr}(u)\) is an initial
state of \(\cH^{r}(\qr, L)\) if{}f \(\cl_{\ql}(u^R)\) is a
final state of \(\cH^{\ell}(\ql, L^R)\), i.e.\ an initial state of \((\cH^{\ell}(\ql, L^R))^R\).
Therefore, \(\varphi(I) = \wt{I}\).

Since \(\cl_{\qr}(u) \subseteq L \Lra u \in L \Lra u^r \in L^R\), we have that \(\cl_{\qr}(u)\) is a final state of \(\cH^{r}(\qr, L)\) if{}f \(\cl_{\ql}(u^R)\) is an initial state of \(\cH^{\ell}(\ql, L^R)\), i.e.\ a final state of \((\cH^{\ell}(\ql, L^R))^R\).
Therefore, \(\varphi(F) = \wt{F}\).

It remains to show that \(q' \in \delta(q, a)\Lra \varphi(q) \in \wt{\delta}(\varphi(q'),a)\), for all \(q, q' \in Q\) and \(a \in \Sigma\).
Assume that \(q = \cl_{\qr}(u)\) for some \(u \in \Sigma^*\), \(q' = \cl_{\qr}(v)\) for some \(v \in Σ^*\) and \(q' \in \delta(q, a)\) with \(a \in \Sigma\).
Then,
\begin{align*}
\cl_{\qr}(v) \in δ(\cl_{\qr}(u), a) & \Lra \quad \text{[By Definition~\ref{def:right-const:qo}]} \\
\cl_{\qr}(u)a \subseteq \cl_{\qr}(v) & \Lra \quad \text{[By definition of \(\cl_{\qr}\) and Lemma~\ref{lemma:QObwComplete}]} \\
v \qr ua & \Lra \quad \text{[Since \(u \qr v \Lra u^R \ql v^R\) and \((ua)^R = au^R\)]} \\
v^r \ql au^R & \Lra \quad \text{[By definition of \(\cl_{\ql}\) and Lemma~\ref{lemma:QObwComplete}]} \\
a\cl_{\ql}(u^R) \subseteq \cl_{\ql}(v^R) & \Lra \quad \text{[By Definition~\ref{def:left-const:qo} and reverse automata]} \\
\cl_{\ql}(u^R) \in \wt{δ}(\cl_{\ql}(v^R), a) & \Lra \quad \text{[Definition of \(q, q'\) and \(\varphi\)]} \\
\varphi(q) \in \wt{δ}(\varphi(q'), a) \enspace .\tag*{\qedhere}
\end{align*}
\end{proof}

\numLPrimePrincipals*
\begin{proof}
We proceed by showing that for every \(L\)-prime \(\cl_{\qo_2}(u)\) there exists an \(L\)-prime \(\cl_{\qo_1}(x)\) such that \(\cl_{\qo_2}(x) = \cl_{\qo_2}(u)\).
Clearly, this entails that there are, at least, as many \(L\)-prime principals for \(\qo_1\) as there are for \(\qo_2\).

Let $\cl_{\qo_2}(u)$ be $L$-prime.

If \(\cl_{\qo_1}(u)\) is \(L\)-prime, we are done.
Otherwise, by Lemma~\ref{lemma:numprincipals}, we have that there exists $x \prec_1 u$ such that 
$\cl_{\qo_2}(u) = \cl_{\qo_2}(x)$.

We repeat the reasoning with $x$. If $\cl_{\qo_1}(x)$ is $L$-prime, we are done. Otherwise, there exists $x_1 \prec_1 x$ such that 
$\cl_{\qo_2}(u) = \cl_{\qo_2}(x) = \cl_{\qo_2}(x_1)$.

Since \(\qo_1\) induces finitely many principals, there are no infinite strictly descending chains and, therefore, there exists \(x_n\) such that \(\cl_{\qo_2}(u)=\cl_{\qo_2}(x)=\cl_{\qo_2}(x_1)=\ldots = \cl_{\qo_2}(x_n)\) and \(\cl_{\qo_1}(x_n)\) is \(L\)-prime.%
\end{proof}

\positiveAtoms*
\begin{proof}
We prove the lemma for the principals induced by \(\qrN\) and \(\qrL\).
The proofs for the left quasiorders are symmetric. 

For each \(u \in \Sigma^*\) we have that 
\begin{align*}
\cl_{\qrN}(u) &= \quad \text{[By definition of \(\cl_{\qrN}\)]}\\
\{v \in Σ^* \mid \post_u^{\cN}(I) \subseteq \post_v^{\cN}(I)\} & = \quad \text{[By definition of set inclusion]} \\
\{v \in \Sigma^* \mid  \forall q \in \post^{\cN}_u(I), q\in \post^{\cN}_v(I)\} & = \quad \text{[Since \(q \in \post^{\cN}_v(I) \Lra v \in W^{\cN}_{I,q}\)]}\\
\{v \in \Sigma^* \mid  \forall q \in \post^{\cN}_u(I), v \in W^{\cN}_{I,q}\} & = \quad \text{[By definition of intersection]}\\
\bigcap\textstyle{_{q \in \post^{\cN}_u(I)}} W_{I,q}^{\cN} & \enspace .
\end{align*}

On the other hand, 
\begin{align*}
v \in  \bigcap\textstyle{_{w \in \Sigma^*, \; w \in u^{-1}L}} Lw^{-1} & \Lra \quad \text{[By definition of intersection]} \\
\forall w \in Σ^*, \; w \in u^{-1}L \Ra v \in Lw^{-1} & \Lra \quad \text{[Since \(\forall x,y \in Σ^*, \; x \in Ly^{-1} \Lra y \in x^{-1}L\)]} \\
\forall w \in Σ^*, \; w \in u^{-1}L \Ra w \in v^{-1}L & \Lra \quad \text{[By definition of set inclusion]} \\
u^{-1}L \subseteq v^{-1}L & \Lra \quad \text{[By definition of \(\cl_{\qrL}(u)\)]} \\
v \in \cl_{\qrL}(u) \tag*{\qedhere}
\end{align*}
\end{proof}

\coResidualqrLqrN*
\begin{proof}
We have that \(\post_u^{\cN}(I) \subseteq \post_v^{\cN}(I) \Ra W_{\post_{u}^{\cN}(I),F}^{\cN} \subseteq W_{\post_{v}^{\cN}(I),F}^{\cN}\) holds for every NFA \(\cN\) and \(u, v \in Σ^*\).
Next we show that the reverse implication holds.
Let \(u, v\in Σ^*\) be such that \(W_{\post_{u}^{\cN}(I),F}^{\cN} \subseteq W_{\post_{v}^{\cN}(I),F}^{\cN}\).
Then,
\begin{align*}
q \in \post_{u}^{\cN}(I) & \Ra \quad \text{[Since \(\cN\) is co-RFA with no empty states]}\\
\exists x \in Σ^*, \; u \in W_{I,q} = Lx^{-1} & \Ra \quad \text{[Since \(u \in Lx^{-1} \Ra x \in u^{-1}L\)]} \\
x \in W_{\post_{u}^{\cN}(I), F} & \Ra \quad \text{[Since \(W_{\post_{u}^{\cN}(I),F}^{\cN} \subseteq W_{\post_{v}^{\cN}(I),F}^{\cN}\)]} \\
x \in W_{\post_{v}^{\cN}(I), F} & \Ra \quad \text{[By definition of \(W_{S,T}^{\cN}\)]} \\
\exists q' \in Q, \; x \in W_{q',F} \land v \in W_{I, q'} & \Ra \quad \text{[Since \(x \in W_{q',F} \Ra W_{I,q'} \subseteq Lx^{-1}\)]}\\
v \in  Lx^{-1} & \Ra \quad \text{[Since \(Lx^{-1} = W_{I,q}\)]} \\
v \in W_{I, q} & \Ra \quad \text{[By definition of \(\post_v^{\cN}(I)\)]} \\
q \in \post_v^{\cN}(I) \enspace .
\end{align*}
Therefore, \(W_{\post_{u}^{\cN}(I),F}^{\cN} \subseteq W_{\post_{v}^{\cN}(I),F}^{\cN} \Ra \post_u^{\cN}(I) \subseteq \post_v^{\cN}(I)\).

The proof for RFAs with no unrechable states and left quasiorders is symmetric.
\end{proof}

\theoremF*
\begin{proof}\hfill
\begin{alphaenumerate}
\item \(\lang{\cF{r}(L)} = \lang{\cF{\ell}(L)} = L = \lang{\cG{r}(\cN)} = \lang{\cG{\ell}(\cN)}\).%

By Definition~\ref{def:FG}, \(\cF{r}(L) = \cH^{r}(\qrL, L)\) and \(\cG{r}(\cN) = \cH^{r}(\qrN, L)\).
By Lemma \ref{lemma: HrGeneratesL}, \(\lang{\cH^{r}(\qrL, L)} = \lang{\cH^{r}(\qrN, L)} = L\).
Therefore, \(\lang{\cF{r}(L)} = \lang{\cG{r}(L)} = L\).
Similarly, it follows from Lemma~\ref{lemma:HlgeneratesL} that \(\lang{\cF{\ell}(L)} = \lang{\cG{\ell}(L)} =L\).

\item \(\cF{\ell}(L)\) is isomorphic to \((\cF{r}(L^R))^R\).

For every \(u,v \in \Sigma^*\):
\begin{align*}
u \qlL v & \Lra \quad \text{[By Definition~\ref{eq:Llanguage}]} \\
u^{-1}L \subseteq v^{-1}L & \Lra \quad \text{[\(A \subseteq B\Lra A^R \subseteq B^R\)]}\\
(u^{-1}L)^R \subseteq (v^{-1}L)^R & \Lra \quad\text{[Since \((u^{-1}L)^R = L^R(u^R)^{-1} \)]} \\
L^R(u^R)^{-1} \subseteq L^R(v^R)^{-1} & \Lra\quad\text{[By Definition~\ref{def:NerodeQO}]} \\
u^R \qr_{L^R} v^R \enspace . 
\end{align*}
Finally, it follows from Lemma~\ref{lemma:leftRightReverse} that \(\cF{\ell}(L)\) is isomorphic to \((\cF{r}(L^R))^R\).

\item \(\cG{\ell}(\cN)\) is isomorphic to \((\cG{r}(\cN^R))^R\).

For every \(u,v \in \Sigma^*\):
\begin{align*}
u \qlN v & \Lra \quad\text{[By Defintion~\ref{def:automataQO}]}\\
\pre_u^{\cN^R}(F) \subseteq \pre_v^{\cN^R}(F) & \Lra\quad \text{[Since \(q \in \pre^{\cN^R}_{x}(F)\) if{}f \(q \in \post^{\cN}_{x^R}(I) \)]}\\
\post_{u^R}^{\cN}(I) \subseteq \post_{v^R}^{\cN}(I) & \Lra \quad\text{[By Definition~\ref{def:automataQO}]}\\
u^R \qlN v^R \enspace .
\end{align*}
It follows from Lemma~\ref{lemma:leftRightReverse} that \(\cG{\ell}(\cN)\) is isomorphic to \(\cG{r}(\cN^R)^R\).

\item \(\cF{r}(L)\) is isomorphic to the canonical RFA for \(L\).

Let \(\cF{r}(L) = (Q, \Sigma, \delta, I, F)\)  and let \(\cC = (\widetilde{Q}, \Sigma, \eta, \widetilde{I}, \widetilde{F})\) be the canonical RFA for \(L\).
Let \(\varphi: \widetilde{Q} \rightarrow Q\) be the mapping assigning to each state \(\widetilde{q}_i \in \widetilde{Q}\) of the form \(u^{-1}L\), the state \(\cl_{\qrL}(u) \in Q\), with \(u \in \Sigma^*\).
We show that \(\varphi\) is an NFA isomorphism between \(\cC\) and \(\cF{r}(L)\).

Since \(u^{-1}L \subseteq L \Lra u \qrL \varepsilon \Lra \varepsilon \in \cl_{\qrL}(u)\), the initial states \(u^{-1}L \in \widetilde{I}\) are mapped to initial states \(\cl_{\qrL}(u)\) of \(\cC\).
Therefore, \(\varphi(\widetilde{I}) = I\).

On the other hand, since \(\varepsilon \in u^{-1}L \Lra u \in L\), each final state \(u^{-1}L \in \widetilde{F}\) is mapped to a final state \(\cl_{\qrL}(u)\) of \(\cC\).
Therefore, \(\varphi(\widetilde{F}) = F\)

Since \(\cl_{\qrL}(u) a \subseteq \cl_{\qrL}(v) \Lra v \qrL ua \Lra v^{-1}L \subseteq (ua)^{-1}L\), it is straightforward to check that \(v^{-1}L = \eta(u^{-1}L, a)\) if and only if \(\cl_{\qrL}(v) \in δ(\cl_{\qrL}(u),a)\), for all \(u^{-1}L, v^{-1}L \in \wt{q}\) and \(a \in \Sigma\).

Finally, we need to show that \(\forall u \in Σ^*, \;  \cl_{\qrL}(u) \in Q \Lra \exists q_i \in \widetilde{Q}, \; q_i = u^{-1}L\).
Observe that:
\begin{align*}
u^{-1}L = \bigcup_{x \qrn_L u} u^{-1}L & \Lra \quad \text{[By Definition~\ref{def:NerodeQO}]} \\
u^{-1}L = \bigcup_{x^{-1}L \subsetneq u^{-1}L} x^{-1}L \enspace .
\end{align*}
It follows that \(\forall u \in Σ^*, \cl_{\qrL}(u) \text{ is \(L\)-prime} \Lra u^{-1}L \text{ is prime}\) and, therefore, \(\varphi(\wt{Q}) = Q\).

\item \(\cG{r}(\cN)\) is isomorphic to a sub-automaton of \(\cN^{\text{res}}\) and \(\lang{\cG{r}(\cN)} = \lang{\cN^{\text{res}}}\).

Given \(\cN = (Q, \Sigma, \delta, I, F)\), recall that \(\cN^{\text{res}} = (Q_r, Σ, δ_r, I_r, F_r)\) is the RFA built by the residualization operation defined by Denis et al.~\cite{denis2002residual}.
Let \(\cG{r}(\cN) = (\widetilde{Q}, \Sigma, \widetilde{\delta}, \widetilde{I}, \widetilde{F})\).

We will show that there is a surjective mapping \(\varphi\) that associates states and transitions of \(\cG{r}(\cN)\) with states and transitions of \(\cN^{\text{res}}\).
Moreover, if \(q \in \widetilde{Q}\) is initial (resp.\ final) then \(\varphi(q) \in Q_r\) is initial (resp.\ final) and \(q' \in \wt{δ}(q,a) \Lra \varphi(q') \in δ_r(\varphi(q),a)\).
In this way, we conclude that \(\cG{r}(\cN)\) is isomorphic to a sub-automaton of \(\cN^{\text{res}}\).
Finally, since \(\lang{\cN^{\text{res}}} = \lang{\cN}\) then it follows from Lemma~\ref{lemma: HrGeneratesL} that \(\lang{\cN^{\text{res}}} = \lang{\cN} = \lang{\cG{r}(\cN)}\).

Let \(\varphi: \widetilde{Q} \rightarrow Q_{r}\) be the mapping assigning to each state \(\cl_{\qrN}(u) \in \widetilde{Q}\)  with \(u \in \Sigma^*\), the set \(\post_u^{\cN}(I) \in Q_{r}\).

It is straightforward to check that the initial states \(\widetilde{I} = \{\cl_{\qrN}(u) \in \widetilde{Q} \mid \varepsilon \in \cl_{\qrN}(u)\}\) of \(\cG{r}(\cN)\) are mapped into the set \(\{\post_u^{\cN}(I) \mid \post_u^{\cN}(I) \subseteq \post_\varepsilon^{\cN}(I)\}\) which are the initial states of \(\cN^{\text{res}}\).

Similarly, each final state of \(\cG{r}(\cN)\), \(\cl_{\qrN}(u)\) with \(u \in \lang{\cN}\), is mapped to \(\post_u^{\cN}(I)\) such that \(\post_u^{\cN}(I) \cap F \neq \varnothing\), hence, \(\post_u^{\cN}(I)\) is a final state of \(\cN^{\text{res}}\).

Moreover, since \(\cl_{\qrN}(u) a \subseteq \cl_{\qrN}(v) \Lra v \qrN ua \Lra \post_{v}^{\cN}(I) \subseteq \post_{ua}^{\cN}(I)\), it follows that \(\forall u, v \in Σ^*\) such that \(\post_u^{\cN}(I), \post_v^{\cN}(I) \in Q_r\), we have \(\post_v^{\cN}(I) \in δ_r(\post_u^{\cN}(I),a) \Lra \cl_{\qrN}(v) \in \wt{δ}(\cl_{\qrN}(u), a)\).

Finally, we show that \(\forall u \in Σ^*, \;  \cl_{\qrN}(u) \in \widetilde{Q} \Ra \post_u^{\cN}(I) \in Q_r\).
By definition of \(\widetilde{Q}\) and \(Q_r\), this is equivalent to showing that for every word \(u \in Σ^*\), if \(\post_u^{\cN}(I)\) is coverable then \(\cl_{\qrN}(u)\) is \(L\)-composite.
Observe that:
\begin{align*}
\post_u^{\cN}(I) = \hspace{-10pt}\bigcup_{\post_x^{\cN}(I) \subsetneq \post_u^{\cN}(I)}\hspace{-10pt} \post_x^{\cN}(I) & \Lra \quad \text{[\(x \qrn_{\mathcal{N}} u \Lra \post_x^{\cN}(I) \subsetneq \post_u^{\cN}(I)\)]} \\
\post_u^{\cN}(I) = \bigcup_{x \qrn_{\mathcal{N}} u} \post_x^{\cN}(I) & \Ra \quad \text{[Since \(W_{\post_u^{\cN}(I),F}^{\cN} = u^{-1}L\)]} \\
u^{-1}L = \bigcup_{x \qrn_{\mathcal{N}} u} x^{-1}L \enspace .
\end{align*}
It follows that if \(\post_u^{\cN}(I)\) is coverable then \(\cl_{\qrN}(u)\) is \(L\)-composite, hence \(\varphi(\wt{Q}) \subseteq Q_r\).

\item \(\cG{r}(\cG{\ell}(\cN))\) is isomorphic to \(\cF{r}(L)\).

By Lemma~\ref{lemma:HlgeneratesL}, \(\cG{\ell}(\cN)\) is a co-RFA  accepting the language \(L\) with no empty states hence, by Lemma~\ref{lemma:coResidual_qrL=qrN}, \(\cG{r}(\cG{\ell}(\cN))\) is isomorphic to \(\cF{r}(\lang{\cG{\ell}(\cN)})=\cF{r}(\lang{\cN})\).\qedhere
\end{alphaenumerate}
\end{proof}

\qrlqrNCanRes*
\begin{proof}

As shown by Theorem~\ref{theoremF}~(\ref{theorem:CanonicalRFAlanguage}), \(\cF{r}(L)\) is the canonical RFA for \(L\), hence it is strongly consistent and, by Lemma~\ref{lemma:qrHEqualqrifHsc}, we have that \(\mathord{\qr_{\cF{r}(L)}} = \mathord{\qrL}\).
On the other hand, if \(\cG{r}(\cN)\) is isomorphic to \(\cF{r}(L)\) we have that \(\mathord{\qr_{\cG{r}(\cN)}} = \mathord{\qr_{\cF{r}(L)}}\), and by Lemma~\ref{lemma:qrHEqualqrifHsc}, \(\mathord{\qr_{\cG{r}(\cN)}} =~ \qrN\).
It follows that if \(\cG{r}(\cN)\) is isomorphic to \(\cF{r}(L)\) then \(\mathord{\qrL} = \mathord{\qrN}\).

Finally, if \(\mathord{\qrL} = \mathord{\qrN}\) then \(\cH^{r}(\qrL, L) = \cH^{r}(\qrN, \lang{\cN})\), i.e.,  \( \cF{r}(L) = \cG{r}(\cN)\).
\end{proof}

\NLqo*
\begin{proof}
Let \(\Pref, \Suf \subseteq Σ^*\) be a prefix-closed and a suffix-closed finite set, respectively, and let \(\cT = (T, \Pref, \Suf)\) be the table built by algorithm NL\(^*\).
Observe that for every \(u,v \in \Pref\):
\begin{align}
u \qA v & \Lra \quad \text{[By Definition~\ref{def:finiteNerode}]} \nonumber\\
{u}^{-1}L \subseteq_{\Suf} {v}^{-1}L & \Lra \quad \text{[By definition of quotient w.r.t \(S\)]} \nonumber\\
\forall x \in S, \; ux \in L \Ra vx \in L & \Lra \quad \text{[By definition of \(\cT\)]} \nonumber\\
\forall x \in S, \; (\row(u)(x) = {+}) \Ra (\row(v)(x) = {+}) & \Lra \quad \text{[By Definition~\ref{def:coverRow}]} \nonumber\\
\row(u) \sqsubseteq \row(v) \enspace .
\label{eq:QOIffRowsSubset}
\end{align}

Moreover, for every \(u,v \in \Pref\) we have that \({u}^{-1}L =_{\Suf} {v}^{-1}L\) if{}f \(\row(u) = \row(v)\).

Next, we show that the join operator applied to rows corresponds to the set union applied to quotients w.r.t \(S\).
Let \(u,v \in \Pref\) and let \(x \in \Suf\).
Then,
\begin{align}
(\row(u) \sqcup \row(v))(x) = {+} & \Lra \quad \text{[By Definition~\ref{def:join}]} \nonumber\\
(\row(u)(x) = {+}) \lor (\row(v)(x) = {+}) & \Lra \quad \text{[By definition of row]} \nonumber\\
(ux \in L )\lor (vx \in L) & \Lra \quad \text{[By definition of quotient w.r.t \(\Suf\)]} \nonumber\\
(x \in {u}^{-1}L) \lor (x \in {v}^{-1}L )& \Lra \quad \text{[By definition of \(\cup\)]}\nonumber \\ 
x \in {u}^{-1}L \cup {v}^{-1}L \enspace .
\label{eq:joinUnion}
\end{align}
Therefore, we can prove that \(\row(u)\) is \(\cT\)-\emph{prime} if{}f \(\cl_{\qA}(u)\) is \(L_{\Suf}\)-prime w.r.t. \(\Pref\).
\begin{align*}
\row(u) = {\textstyle\bigsqcup_{v \in \Pref, \; \row(v) \sqsubset \row(u)}} \row(v) & \Lra \quad \text{[By Equation~\eqref{eq:QOIffRowsSubset}]} \\
\row(u) = {\textstyle\bigsqcup_{v \in \Pref, \; {v}^{-1}L \subsetneq_{\Suf} {u}^{-1}L}} \row(v) & \Lra \quad \text{[By Equation~\eqref{eq:joinUnion}]} \\
{u}^{-1}L =_{\Suf} {\textstyle\bigcup_{v \in \Pref, \; {v}^{-1}L \subsetneq_{\Suf} {u}^{-1}L}} {v}^{-1}L & \Lra \quad \text{[\({v}^{-1}L \subsetneq_{\Suf} {u}^{-1}L \Lra u \qAn v\)]} \\
{u}^{-1}L =_{\Suf} {\textstyle\bigcup_{v \in \Pref, \; u \qAn v}} {v}^{-1}L  \enspace .
\end{align*}

It follows from Definitions~\ref{def:ClosedCons} (\ref{def:ClosedCons:Closed}) and~\ref{def:Table} (\ref{def:Table:closed}) and Equation~\eqref{eq:joinUnion} that \(\cT\) is closed if{}f \(\qA\) is closed.
Moreover, it follows from Definitions~\ref{def:ClosedCons} (\ref{def:ClosedCons:Cons}) and~\ref{def:Table} (\ref{def:Table:Consistent}) that \(\cT\) is consistent if{}f \(\qA\) is consistent.

On the other hand, for every \(u,v \in \Pref, a \in Σ\) and \(x \in \Suf\) we have that:
\begin{align*}
(\row(u) \subseteq \row(v)) \land (\row(ua)(x) = {+}) \land (\row(va)(x) = {-} )& \Lra \quad \text{[By Equation~\eqref{eq:QOIffRowsSubset}]} \\
(u \qA v )\land (uax \in L) \land (vax \notin L) & %
\end{align*}
It follows that if \(\cT\) and \(\qA\) are not consistent then both NL\(^*\) and NL\(^{\qo}\) can find the same word \(ax \in Σ   \Suf\) and add it to \(\Suf\).
Similarly, it is straightforward to check that if \(\row(ua)\) with \(u \in \Pref\) and \(a \in Σ\) break consistency, i.e.\ it is \(\cT\)-prime and it is not equal to any \(\row(v)\) with \(v \in \Pref\), then \(\cl_{\qA}(ua)\) is \(L_{\Suf}\)-prime for \(\Pref\) and not equal to any \(\cl_{\qA}(v)\) with \(v \in \Pref\).
Thus, if \(\cT\) and \(\qA\) are not closed then both NL\(^*\) and NL\(^{\qo}\) can find the same word \(ua\) and add it to \(\Pref\).

It remains to show that both algorithms build the same automaton modulo isomorphism, i.e., \(\cR(\cT) = (\widetilde{Q}, Σ, \wt(δ), \widetilde{I}, \wt{F})\) is isomorphic to \(\cL(\qA, \Pref) = (Q, Σ, δ, I, F)\).
Define the mapping \(\varphi: Q \to \wt{Q}\) as \(\varphi(\cl_{\qA}(u)) = \row(u)\).
Then:
\begin{align*}
\varphi(Q) & = \{\varphi(\cl_{\qA}(u)) \mid u \in \cP \land \cl_{\qA}(u) \text{ is \(L_{\Suf}\)-prime w.r.t. \(\Pref\)}\} \\
& = \{\row(u) \mid u \in \cP \land \row(u) \text{ is \(\cT\)-prime}\} = \wt{Q} \enspace .\\
\varphi(I) & = \{\varphi(\cl_{\qA}(u)) \mid \varepsilon \in \cl_{\qA}(u)\} = \{\row(u) \mid u \qA \varepsilon\} \\
& = \{\row(u) \mid \row(u) \sqsubseteq \row(\varepsilon)\} = \wt{I} \enspace .\\
\varphi(F) & = \{\varphi(\cl_{\qA}(u)) \mid u \in L \cap \cP\} = \{\row(u) \mid u \in L \cap \cP\} \\
& = \{\row(u) \mid \row(u)(\varepsilon) = {+}\} = \wt{F}\enspace .\\
\varphi(\delta(\cl_{\qA}(u),a)) &= \varphi(\cl_{\qA}(ua)) = \{\row(v) \mid \cl_{\qA}(u) \in Q \land \cl_{\qA}(u)a \subseteq \cl_{\qA}(v)\} \\
& = \{\row(v) \mid \row(v) \in \wt{Q} \land v \qA ua\}  \\
& = \{\row(v) \mid \row(v) \in \wt{Q} \land \row(v) \sqsubseteq \row(ua)\} \\
& = \wt{\delta}(\row(u),a) = \wt{\delta}(\varphi(\cl_{\qA}(u)),a) \enspace .
\end{align*}

Finally, we show that \(\varphi\) is an isomorphism.
Clearly, the function \(\varphi\) is surjective since, for every \(u \in \Pref\), we have that \(\row(u) = \varphi(\cl_{\qA}(u))\).
Moreover \(\varphi\) is injective since for every \(u,v \in \Pref\), \(\row(u) = \row(v) \Lra {u}^{-1}L =_{\Suf} {v}^{-1}L\), hence \(\row(u) = \row(v) \Lra \cl_{\qA}(u) = \cl_{\qA}(v)\).

We conclude that \(\varphi\) is an NFA isomorphism between \(\cL(\qA,\Pref))\) and \(\cR(\cT)\).
Therefore NL\(^*\) and NL\(^{\qo}\) exhibit the same behavior, provided that both algorithms resolve nondeterminism in the same way, as they both maintain the same sets \(\Pref\) and \(\Suf\) and build the same automata at each step.
\end{proof}

\end{document}